\documentclass[11pt]{article}

\usepackage{parskip}
\usepackage{amsmath}
\usepackage{amssymb}
\usepackage{xspace}
\usepackage{tikz}
\usepackage{a4wide}

\usepackage{algorithm}
\usepackage[noend]{algpseudocode}
\usepackage{xparse} % Used when adding to the environment
\newtheorem{xdefinition}{Definition}
\newtheorem{xobservation}{Observation}
\newtheorem{xtheorem}{Theorem}
\newtheorem{xlemma}{Lemma}
\newtheorem{xproposition}{Proposition}
\newtheorem{xcorollary}{Corollary}
\newenvironment{definition}{\begin{xdefinition}\rm}%
{\hspace*{\fill}\raisebox{-1pt}{\boldmath$\Box$}\end{xdefinition}}
{\hspace*{\fill}\raisebox{-1pt}{\boldmath$\Box$}\end{xobservation}}
\newenvironment{theorem}{\begin{xtheorem}\rm}{\end{xtheorem}}
\newenvironment{lemma}{\begin{xlemma}\rm}{\end{xlemma}}

\newenvironment{corollary}{\begin{xcorollary}\rm}{\end{xcorollary}}
\newenvironment{proof}{\begin{trivlist}\item[]{\bf Proof }}%
{\hspace*{\fill}\raisebox{-1pt}{\boldmath$\Box$}\end{trivlist}}

\newcommand{\OPT}{\ensuremath{\operatorname{\textsc{Opt}}}\xspace}
\newcommand{\ALG}{\ensuremath{\operatorname{\textsc{Alg}}}\xspace}
\newcommand{\ADV}{\ensuremath{\operatorname{\textsc{Adv}}}\xspace}
\newcommand{\ALGp}{\ensuremath{\operatorname{\textsc{Alg}}'}\xspace}
\newcommand{\ADVp}{\ensuremath{\operatorname{\textsc{Adv}}'}\xspace}
\newcommand{\mist}{\ensuremath{\operatorname{\texttt{mis3}}}\xspace}
\newcommand{\OBJ}{\ensuremath{\operatorname{\textsc{Obj}}}\xspace}

\newcommand{\VC}{\ensuremath{\operatorname{\textsc{PriorityVC}}}\xspace}
\newcommand{\SET}[1]{\left\{#1\right\}}
\newcommand{\Oh}[1]{\ensuremath{O^*(#1)}}
\newcommand{\lOh}[1]{\ensuremath{O^*\left(#1\right)}}
\newcommand{\BAD}{\ensuremath{\operatorname{\textsc{Bad}}}\xspace}
\newcommand{\first}[1]{\textit{first}(#1)}
\newcommand{\EPSINT}{\ensuremath{(0,\frac12]}}
\newcommand{\FLOOR}[1]{\left\lfloor#1\right\rfloor}
\newcommand{\SIZE}[1]{\left|#1\right|}
\newcommand{\SETOF}[2]{\SET{#1 \mid #2}}

\DeclareMathOperator*{\argmax}{arg\,max}
\DeclareMathOperator{\poly}{poly}

\algnewcommand\algorithmicswitch{\textbf{switch}}
\algnewcommand\algorithmiccase{\textbf{case}}
\algdef{SE}[SWITCH]{Switch}{EndSwitch}[1]{\algorithmicswitch\ #1\ %\algorithmicdo
}{\algorithmicend\ \algorithmicswitch}
\algdef{SE}[CASE]{Case}{EndCase}[1]{\algorithmiccase\ #1}{\algorithmicend\ \algorithmiccase}%
\algtext*{EndSwitch}%
\algtext*{EndCase}%
\begin{document}

\title{Advice Complexity of Adaptive Priority Algorithms}

\author{Joan Boyar\thanks{Department of Mathematics and Computer Science, University of Southern Denmark, Denmark. Email: joan@imada.sdu.dk}
  \and Kim S. Larsen\thanks{Department of Mathematics and Computer Science, University of Southern Denmark, Denmark. Email: kslarsen@imada.sdu.dk}
  \and Denis Pankratov\thanks{Department of Computer Science and Software Engineering, Concordia University, Quebec, Canada. Email: denis.pankratov@concordia.ca}}

\maketitle

\abstract{The priority model was introduced to capture ``greedy-like''
  algorithms. Motivated by the success of advice complexity in the
  area of online algorithms, the \emph{fixed} priority model was
  extended to include advice, and a reduction-based framework was
  developed for proving lower bounds on the amount of advice required
  to achieve certain approximation ratios in this rather powerful
  model. To capture most of the algorithms that are considered
  greedy-like, the even stronger model of \emph{adaptive} priority
  algorithms is needed.  We extend the adaptive priority model to
  include advice. We modify the reduction-based framework from the
  fixed priority case to work with the more powerful adaptive priority
  algorithms, simplifying the proof of correctness and strengthening
  all previous lower bounds by a factor of two in the process.

  We also present a purely combinatorial adaptive priority algorithm
  with advice for Minimum Vertex Cover on triangle-free graphs of
  maximum degree three. Our algorithm achieves optimality and uses at
  most $7n/22$ bits of advice.  No adaptive priority algorithm without
  advice can achieve optimality without advice, and we prove that an
  online algorithm with advice needs more than $7n/22$ bits of advice to
  reach optimality.
  
  We show connections between exact algorithms and priority algorithms
  with advice. The branching in branch-and-reduce algorithms can be
  seen as trying all possible advice strings, and all priority
  algorithms with advice that achieve optimality define corresponding
  exact algorithms, \emph{priority exact algorithms}.  Lower bounds on
  advice-based adaptive algorithms imply lower bounds on running times
  of exact algorithms designed in this way.  }

%\keywords{greedy algorithms, priority algorithms, adaptive priority algorithms, exact algorithms, online algorithms, advice complexity}

\section{Introduction}
\label{sec:intro}
Everybody who has studied algorithms has an intuitive notion of
a greedy algorithm.  In many discrete optimization problems, input can
be represented as a sequence of items coming from some infinite
universe, and the output of an algorithm can be represented as a
sequence of decisions -- one decision per item.  A decision could, for
example, be to accept or reject an item.  The quality of such a
sequence of decisions is often measured using an objective function
that must be maximized (or minimized).  Greediness refers to making
the decision that maximizes the objective function at this point.
This often means that the algorithm pretends that each input item is
the last it is going to receive.\footnote{For some problems, in
  particular many graph problems, the input items received so far may
  require a certain number of further input items to be given before a
  well-defined final input is formed; see~\cite{BorodinBLM2010} for a
  detailed discussion of these issues.}

One of the earliest formalizations of a greedy-like notion was in the
form of matroids by Whitney~\cite{Whitney35}, more recently extended
to greedoids by Korte and
Lov\'{a}sz~\cite{KorteL1981,KorteL1983,KorteL1984,KorteL19842}.  In
spite of the profound connection between greedoids and optimization
problems admitting optimal greedy algorithms, greedoids do not give a
complete characterization of what people usually characterize as
greedy algorithms, and there is no consensus in the research community
as to a formal definition of greedy algorithms.

Priority algorithms were introduced by Borodin, Nielsen, and
Rackoff~\cite{BorodinNR2003} in an attempt to formalize
``greedy-like'' or ``myopic'' algorithms, trying to encompass the
algorithm designers' notion of greedy-like that goes beyond the
matroid-based framework (earlier works such as
\cite{Graham1966,KargerSW99} have discussed the basic idea of using
priority functions for scheduling problems as an informal but fairly
well understood concept). One of the purposes of this formalization is
to prove results giving lower bounds on how well any priority
algorithm can approximate, without requiring any assumptions such as
$\textrm{P}\not=\textrm{NP}$.  The priority model has been studied in
the context of many combinatorial optimization topics, including
classical graph
problems~\cite{AngelopoulosB2004,DavisI2009,BorodinBLM2010,BesserP17},
scheduling~\cite{BorodinNR2003,Regev02,LiptonT94,Papakonstantinou06},
satisfiability~\cite{Poloczek11,PoloczekSWZ17},
auctions~\cite{BorodinL16}, and general results, present in many of
the above contributions as well as in~\cite{LeshM06}. Many classical
greedy algorithms have a simple structure consisting of two
components: a sorting, ordering, or priority component and an online,
irrevocable decisions component. The second component is where an
irrevocable decision is made, while the first component determines the
order in which the items are processed by that second component.
Priority algorithms have this structure and they come in two flavors:
fixed and adaptive. We illustrate these models with two well known
examples.

The input to the Minimum Spanning Tree problem is an edge-weighted,
undirected, connected graph, and the objective is to select a set of
edges forming a spanning tree of minimal total weight.  Viewing
Kruskal's algorithm for this problem as a \emph{fixed} priority
algorithm, we define the universe of input items as
$\mathcal{U} = \{(u, v, w) \in
\mathbb{N}\times\mathbb{N}\times\mathbb{Q} \mid u \neq v\}$, where
$(u,v)$ is an edge between vertices $u$ and $v$ with weight~$w$.  An
input instance is a finite subset $\mathcal{I} \subset
\mathcal{U}$. Kruskal's algorithm can be thought of as defining an
ordering on the entire universe $\mathcal{U}$ (by non-decreasing
weight~$w$, with arbitrary tie-breaking) \emph{prior to seeing any
  input items}. The input $\mathcal{I}$ is then given to the algorithm
one input item at a time, in the order defined on the universe.  When
we discuss correctness and quality, we often think of the input being
given by an adversary, but of course still respecting the ordering
that may not be total.  The algorithm makes an \emph{irrevocable}
decision when receiving the next item: accept the edge if it does not
form a cycle with the current partial solution (the set of accepted
items so far), and reject it otherwise.

Strengthening the model, \emph{adaptive} priority algorithms may
change the ordering of the universe after processing each input
item. An example of an adaptive priority algorithm is Prim's algorithm
for the Minimum Spanning Tree problem. The universe is as above.
Prim's algorithm also orders edges by non-decreasing weight, but it
has to maintain a single connected component. Thus, the algorithm
gives higher priority to edges incident to vertices already added to
the solution. Since the set of vertices in the solution keeps growing,
the ordering (the priority function) is updated in every step.  We
emphasize that it is an ordering of the universe, the rest of the
input is not known, and the ordering is redefined \emph{before} the
next input is given.

Note that online algorithms are usually only used when problems have
an online nature, while priority algorithms provide a framework for
certain offline algorithms.  However, as models, they seem quite
similar.  Priority algorithms can be seen as either extending the
power of online algorithms by allowing a limited ordering of input
items, or as limiting the power of an adversary by not allowing it
full control over the order of items.

We now discuss \emph{advice}, starting with the online algorithms setting,
where advice has been considered for some time.  An online algorithm
processes a sequence of input items, one at a time,
with no knowledge of future input items; an assumption that, even
for inherently online problems, is not necessarily realistic.  Often
some information about the input sequence is known in advance, e.g.,
its length, the largest weight of an item, etc.  The knowledge could
be absolute, approximate, or expected from experience. An
information-theoretic way of capturing some of this additional
knowledge is provided by the \emph{advice tape model}\footnote{Other
  advice models have been proposed, including the helper and answerer
  models of Dobrev et al.~\cite{A4}, the tree exploration model with
  advice of Fraigniaud et
  al.~\cite{DBLP:journals/iandc/FraigniaudIP08}, and the per request
  model of Emek et al.~\cite{A2}.  See~\cite{BoyarFKLM2016} for a
  comparison of these models.} of Hromkovi\v{c} et al.~\cite{A3}
(further technical development in B\"{o}ckenhauer et
al.~\cite{BKKKM17}). In this model, an all powerful oracle that knows
the algorithm and sees the entire input sequence\footnote{In contrast
  with the online and priority worlds, in the Turing machine world the
  advice depends only on the input length $n$ and not the input
  itself.} writes bits (referred to as advice bits) on an infinite
tape. The algorithm uses the advice tape in processing the online
items. The ``tape'' analogy is used in many other models, but the only
important properties are that there are always bits when the algorithm
asks for them and there is no detectable end to the collection of
bits.  The \emph{advice complexity} of an algorithm is the number of
bits read. Usually, we are interested in the worst-case number of bits
read as a function of the input length.  Results for online algorithms with
advice are bounds on
the number of advice bits necessary and/or sufficient to achieve a
given competitive ratio\footnote{The competitive ratio is the term
  used in online algorithms for what is essentially the approximation
  ratio when considering offline problems.}.
Often, a few bits of advice improves the competitive ratio dramatically
over what is achievable by an online algorithm without advice.

The lower bound results can be interpreted as hardness results for the
online problems: if many advice bits are necessary in order to reach
optimality (or significantly improve the competitive ratio), the
problem is hard.  Results can also give strong lower bounds on certain
types of semi-online algorithms and inspire algorithm design.
See~\cite{BoyarFKLM2016} for an extensive list of articles. Of most
relevance to us are results concerning graph
algorithms~\cite{DBLP:journals/ijfcs/BianchiBBKP18,Ais,DBLP:conf/sirocco/DobrevKM12,Magnus,GP19,DBLP:journals/corr/KommKKK15,Ahunt,Abipfinal}.

A superset of the current authors introduced advice into the
\emph{fixed} priority model~\cite{BBLP20}.  As for online algorithms
in the advice tape model, an oracle knows the algorithm, sees the
entire input sequence, and writes advice bits on the tape. The advice
is then read by the priority algorithm at its discretion during its
execution.  Just to emphasize, since the oracle knows the algorithm,
the bits always represent what the algorithm expects, so the oracle
and the algorithm cooperate.  In this model, one is interested in the
number of advice bits 
necessary and/or sufficient to achieve a given approximation ratio.
In addition to introducing this model, \cite{BBLP20}~also developed a
general framework for proving lower bounds in this model and applied
this framework to several classical problems, including Maximum
Independent Set, Maximum Bipartite Matching, Minimum Vertex Cover,
etc.  That paper left it as an open question whether the ideas can be
extended to the (arguably more useful) \emph{adaptive} priority model,
and if this would result in useful new paradigms.  Our current paper
addresses that question.

There are many models that represent computation as a leveled tree (or
even more generally as a DAG -- directed acyclic graph), such as
decision trees, branching programs, small depth formulas/circuits,
various proof systems (tree-like and general resolution), pBT
algorithms, etc. One can often define a notion for each of these
tree/DAG models which intuitively captures the amount of parallelism
needed to carry out the computation efficiently. Such a notion can be
viewed as being somewhat analogous to the notion of advice in our
setting. For example, in the pBT (priority backtracking) model of
Alekhnovich et al.~\cite{alekhnovich2005toward}, an algorithm is
represented by a pair of functions: one function allows reordering of
the universe of input items, and another function assigns a value to a
decision based on already seen input items. The ordering function can
be fixed, adaptive, or fully adaptive (we are not discussing this in full here).
The execution of such an algorithm on a particular instance can
be represented by an ordered leveled tree, where each node
corresponds to a partial execution and is labeled by the sequence of
input items seen so far and decisions made for those items. The
children of a node (in order from left to right) correspond to
different input items to be considered next according to the current
ordering function. The correctness condition requires that at least
one of the leaves contains an optimal choice of decisions. The width
of a pBT algorithm is the maximum width of a level of such a tree,
where the maximum is taken over all levels and all instances of a
given length. The length of the ordered (left-to-right) depth-first
search traversal of a the pBT tree corresponds to the running time of
the natural backtracking algorithm associated with the pBT
algorithm. This model captures many backtracking algorithms, but not
all of them. For example, early termination as well as choices of
which decision to make next can be based on only the already seen
portion of the input in pBT, and these choices cannot be made, for
example, based on the value of an LP-relaxation of the entire
instance (as is often done in real-life backtracking algorithms). The
logarithm of the width of a pBT algorithm can be thought of as
``advice'' length, but there are notable differences between the pBT
model and the priority algorithms with advice model. In particular, one
can try to simulate the pBT model by a priority algorithm with advice,
and vice versa, but one quickly runs into issues of whether priorities
and/or decisions are allowed to depend on advice. Establishing precise
connections between these models is an interesting open
problem. Connections between the fixed pBT algorithms and fixed
priority algorithms with advice were previously discussed
in~\cite{BBLP20}. While it is interesting to carry out a comparative
study between various tree/DAG-like models and expose informal and
formal connections between them and the notion of advice, it is not
the goal of the present paper. We discuss only one such connection at
length later in this paper, and that is the connection between
priority algorithms and branch-and-bound/branch-and-reduce algorithms.

%For example, fixed priority algorithms with advice can be viewed in terms of the fixed priority backtracking model of Alekhnovich et al That model starts by ordering the inputs using a fixed priority function and then executes a computation tree where different decisions can be tried for the same input item by branching in the tree, and then choosing the best result. The lower bound results generally consider how much width (maximum number of nodes for any fixed depth in the tree) is necessary to obtain optimality where the width proven is often of the form $2^{\Omega(n)}$. In contrast, fixed priority algorithm with advice results give a parameterized trade-off between the number of advice bits and the competitive ratio. However, given an algorithm in the fixed priority backtracking model, the logarithm of the width gives an upper bound on the number of bits of advice needed for the same approximation ratio. Similarly, a lower bound on the advice complexity gives a lower bound on width. 

We now briefly list our contributions.
\begin{itemize}
\item We introduce the notion of advice in the adaptive priority model
  and identify four natural models based on how the priority function
  is allowed to depend on the advice.
\item We extend the general lower bound framework of~\cite{BBLP20} to
  work in what we call the \emph{oblivious priority function} model.
  The results automatically apply to the weakest model which does not
  use advice in the priority functions at all and also to the fixed
  priority results in~\cite{BBLP20}.  We simplify the proof that the
  framework from~\cite{BBLP20} works, and we strengthen the lower
  bounds implied by the framework by a factor of~$2$.  The framework
  offers a template for lower bound results: By exhibiting %a few
  gadget pattern pairs fulfilling a given list of criteria, a lower bound can be
  computed with fairly limited work.
\item We study the classical Minimum Vertex Cover problem on
  triangle-free graphs of maximum degree~$3$, as a non-trivial example
  problem.  We present an adaptive priority algorithm with advice that
  achieves optimality. The algorithm works in all but the
  weakest of our models.  Known results imply that adaptive priority
  algorithms for this problem cannot achieve optimality without
  advice~\cite{BorodinBLM2010}.  We show that online algorithms must
  use more advice than our algorithm to achieve optimality.  Our
  algorithm is purely combinatorial and requires a somewhat involved
  analysis. This is the most technical of our contributions.
\item Priority algorithms with advice that achieve optimality
  naturally lead to exact algorithms by trying all possible advice
  strings of length no more than the upper bound proven. We call exact
  algorithms designed this way \emph{priority exact algorithms}. We
  discuss the implications of our lower bounds on priority algorithms
  with advice
  for proving lower bounds on the running times of
  such algorithms.
\end{itemize}

In \cite{BBLP20}, the lower bound template is based on an
advice-preserving reduction between two problems within the priority
framework: it is established that if there exists a fixed priority
algorithm with advice for problem $A$, then there also exists one
for Pair Matching (PM) with the same advice length, and it is
shown that PM requires a lot of advice. Such a reduction
must map each input for PM to an
input for $A$, so that decisions for $A$ can
be used for making decisions for PM. The difficulty is that
the inputs and the decisions for $A$ and PM must be aligned
so that inputs respect priority functions, and decisions are not based
on information not available at that point during the
execution of the algorithm.  This becomes significantly harder when
moving to adaptive priority algorithms, since the priorities for the two
problems can depend on advice and can change dramatically between
input items.  We avoid some of these difficulties
by working with an advice-preserving reduction between a problem in an
online setting and a problem in the priority setting, removing the
difficulties in aligning priority functions, and allowing us to focus
more on how priority functions are allowed to depend on advice.
Our extension to \emph{adaptive} priority algorithms enables us to define
and establish lower bounds for priority exact algorithms.

The remainder of the paper is organized as follows.
Section~\ref{sec:models} introduces the four adaptive priority models
with advice.  In Section~\ref{sec:exact}, we discuss connections to
exact algorithms.  In Section~\ref{sec:firstexample}, we show the
first lower bound, based on a construction from~\cite{BorodinBLM2010},
and show that the result is tight for a restricted problem.
This first example problem serves as an introduction to some of
our lower bound techniques. In
Section~\ref{sec:vc}, we present our adaptive priority algorithm for
the Minimum Vertex Cover problem on triangle-free graphs of degree at
most~$3$ and analyze its advice complexity. Section~\ref{sec:template}
presents the extension of the general lower bound framework
of~\cite{BBLP20} to adaptive priority with advice, along with a new
framework for algorithms that solve to optimality. Another
example problem is considered in Section~\ref{sec:tp}, presenting
different lower bounds obtained in two of the different models, along
with a matching upper bound in one of the two models.
Open problems are discussed in Section~\ref{sec:conclusion}.

\section{Models}
\label{sec:models}

A \emph{request-answer game}~\cite{BBKTW94,RaghavanS94} is specified
by the universe of input items $\mathcal{U}$, the universe of
decisions $\mathcal{D}$, the objective function
$\OBJ\colon\mathcal{U}^n \times \mathcal{D}^n \rightarrow \mathbb{R}
\cup \{\pm \infty\}$ on inputs of length $n$, and the type of a
problem, which could be either ``maximization'' or
``minimization''. An input to the request-answer game is a finite
multi-set of items from the universe, i.e., $X=\{x_1, \ldots, x_n\}$
where $x_i \in \mathcal{U}$.  We assume that the objective function is
invariant under simultaneous permutations of input items and
decisions, i.e., for all $x_1, \ldots, x_n$, all $d_1, \ldots, d_n$,
and all permutations $\pi \colon [n] \rightarrow [n]$,
\[ \OBJ(x_1, \ldots, x_n, d_1, \ldots, d_n) = \OBJ(x_{\pi(1)}, \ldots,
  x_{\pi(n)}, d_{\pi(1)}, \ldots, d_{\pi(n)}).\]
The values $\pm \infty$ in the objective can be used to specify
infeasible input. The setting of request-answer games is very general
and includes most problems of interest in the areas of online and
priority algorithms.

A function $P\colon \mathcal{U} \rightarrow \mathbb{R}$ is called a
\emph{priority function}. We introduce a short-hand notation
$\max_P X := \argmax \{P(x) \mid x \in X\}$ for the element of highest
priority in the multi-set $X$. In case there are multiple elements of
highest priority, we assume ties are broken in an adversarial fashion,
i.e., we assume the most unfavorable tie-breaking for our
algorithms. Thus, all upper bounds we prove will be valid for all input
instances, and we can make the simplifying assumption that
$\max_P X$ is an element, and not a set of elements.

A priority algorithm $\ALG$ is not given all of the input, $X$, at
once. Instead, $\ALG$ receives $X$ one item at a time. The priority
algorithm has some limited control over the order in which $X$ is
given: Each time, before the next input item is given, \ALG defines a
priority function, $P$, and the next input item given to \ALG is
$\max_P X$. Recall that the priority function is defined on the
universe, $\mathcal{U}$, and not directly on the remaining part of the
input $X$, which is not known to the algorithm.  We use both the
terminology that an input item has been \emph{given} to the algorithm
and that the algorithm has \emph{received} or \emph{gets} an input item.

What we have described above is the most general version of priority
algorithms, called \emph{adaptive}, since the priority function can be
adapted based on the input given so far. As the name
indicates, \emph{fixed} priority algorithms are those where the
priority function cannot be updated during the execution of the
algorithm.  This simpler class was treated in~\cite{BBLP20}.

We consider priority algorithms in the advice tape
model~\cite{A3,BKKKM17}, and start with a discussion of this model.
The setup is exactly as described for online algorithms in the
introduction.  In the advice tape model, there are two cooperating
players -- the algorithm and the oracle. The oracle sees the entire
input $X$ and writes advice to the algorithm on the infinite advice
tape using the binary alphabet.  The algorithm can decide to read zero
or more bits (for emphasis, often referred to as advice bits) from the
advice tape, sequentially from left to right, before making each
decision. We use $s_i$ to refer to the prefix of the advice tape that
has been read so far by the algorithm. The maximum number of advice
bits read, that is, the largest value of $\lvert s_n \rvert$ for any input of
size~$n$, is the \emph{advice complexity} of the algorithm (a function
of~$n$).  See Algorithm~\ref{alg:priority_template} for a template
illustrating the setup for a priority algorithm with advice.

\begin{algorithm}
\caption{Template: Priority Algorithm with Advice}\label{alg:priority_template}
\begin{algorithmic}[1]
\State $X$ is the input
  \State read zero or more bits from the advice tape
\State $s_0 \gets$ the prefix of the advice string just read
\State $i \gets 1$
\While {$X \not= \emptyset$}
  \State $P_i \gets$ the priority function for iteration~$i$
  \State $x_i \gets \max_{P_i} X$
  \State read zero or more bits from the advice tape
  \State $s_i \gets$ the known content of the advice string
  \State $d_i \gets D_i(x_1, x_2, \ldots, x_i, d_1, d_2, \ldots, d_{i-1}, s_i)$
         -- the decision for input $x_i$
  \State $X \gets X \setminus \{x_i\}$
  \State $i \gets i+1$
\EndWhile
\end{algorithmic}
\end{algorithm}

A priority algorithm with advice must have this format.  A concrete
algorithm is defined by specifying three elements for each iteration:
the priority functions $P_i$, how many advice bits to read, and how
the decision $d_i$ is made.

The decisions, $d_i$, and how many bits of advice to read,
$\lvert s_{i+1} \rvert-\lvert s_i\rvert$, are always functions of the information seen so
far, i.e., the input seen so far, the advice seen so far, and the
previous decisions.  Of course, one may omit the dependence of $d_i$
on $d_1, \ldots, d_{i-1}$, since these decisions can be reconstructed
from $x_1, \ldots, x_{i-1}$ and $s_{i-1}$.  As mentioned in the
introduction, priority algorithms with advice can give rise to
practical algorithms.  However, as a starting point, advice is created
by an oracle, and the setup is used to measure some aspect of problem
difficulty.  Thus, it makes sense to consider how advice may be used
by the algorithm.  In particular, to what extent do we allow the
priority functions to be defined based on the advice obtained by the
algorithm at a given time? We make the following distinctions:

\noindent\textbf{Unrestricted priority function model.} We allow the
priority functions to depend on the input received so far and the
advice read so far:
\[P_i (x_1, \ldots, x_{i-1}, s_{i-1}).\]

\noindent\textbf{Oblivious priority function model.} We allow the priority
function to depend on the input received so far and the advice read so
far, as in the unrestricted priority function model, but the priority
function must give the same priority to all input items which are
indistinguishable, when ignoring names not present in the input items
already seen. (For example, for unweighted graph problems, vertices of
the same degree, where neither the vertices nor their neighbors have
been seen yet, should have the same priority.)

\noindent\textbf{Decision-based priority function model.} We allow the
priority functions to depend on the input received so far and the
decisions made so far:
\[P_i (x_1, \ldots, x_{i-1}, d_1, \ldots, d_{i-1}).\]

\noindent\textbf{Advice-free priority function model.} We only allow the
priority functions to depend on the input received so far:
\[P_i (x_1, \ldots, x_{i-1}).\] Similarly to this, in~\cite{BBLP20}
the priority functions were assumed to not depend on the advice (but
the priority function was fixed, not adaptive).

Clearly, any algorithm that works in the oblivious priority function
model also works in the unrestricted priority function model.  Any
algorithm that works in the decision-based priority function model
also works in the unrestricted priority function model, since the
input and advice determine the decisions. Similarly, any algorithm
that works in the advice-free priority function model can be simulated
by an algorithm in any of the other models, for which reason we refer
to this model as the \emph{weakest}.  Observe that the unrestricted
and decision-based priority functions models coincide when advice
encodes the decisions to be made.  This sometimes functions as a point
of reference, since no more advice than encoding all the decisions can
be necessary. The oblivious priority function model appears to be
incomparable to the decision-based priority function model and its
motivation is as follows.  Although it seems natural to let decisions
depend on the advice in any way and it makes sense to let the priority
function depend on advice, it does not seem natural for an algorithm
to use, for example, a priority function that prefers input items with
certain names that have not been seen yet.

When including advice, one can ask how computationally expensive it is
to generate that advice. This could vary significantly from one
algorithm/application to the next, but the model allows anything; the
priority model does not impose any computational restrictions on
priority functions or decisions by the algorithm.  This is in line
with the information-theoretic nature of the priority model and
similar to other areas, such as online algorithms, communication
complexity, decision tree complexity, etc. These models sidestep hard
computational questions, such as P vs.\ NP, by introducing
informational bottlenecks.  The strengths of this
information-theoretic modeling are that it makes the proven lower
bounds stronger and that it makes it possible to prove results that do
not depend on unproven assumptions in complexity theory. The main
weakness of this information-theoretic modeling is that the algorithms
that are designed might be impractical.  However, priority algorithms
achieving good approximation ratios tend to have easily computable
priority functions and easily computable decisions.

\section{Priority Exact Algorithms}
\label{sec:exact}
There is a simple, general technique one can use to convert a priority
algorithm with advice to an offline algorithm with the same
approximation ratio. If the algorithm uses at most $\ell$ bits of
advice for some input length, then, on an input of that length, one
can enumerate all $2^\ell$ advice strings and execute the algorithm on
each of them, keeping track of the best result.  We call such
algorithms \emph{priority exact algorithms}, since algorithms which
solve problems to optimality are generally referred to as exact
algorithms. 

\subsection{Example: Maximum Independent Set}
\label{misthree}
In the textbook \emph{Exact Exponential Algorithms} by Fomin and
Krasch~\cite{FK10}, in presenting the measure and conquer technique,
they begin with a simple branching algorithm, {\mist}
(Algorithm~\ref{codemis3}), for Maximum Independent Set, the problem
of finding the maximum size among subsets of the vertices where no two
of the vertices are adjacent.
We show how \mist could be changed to a priority exact algorithm for graphs
of bounded degree at most~$\Delta$.

\begin{algorithm}
\caption{Maximum Independent Set algorithm \mist from~\cite{FK10}. $N[v]$ denotes $\{ v\}\cup\{ \mbox{\rm neighbors of }v\}$, $d(v)$ the current degree of~$v$, $\Delta(G)$ the maximum degree in~$G$, and $\alpha(G)$ the size of the maximum independent set.}
\label{codemis3}
\begin{algorithmic}[1]
\State {\bf Algorithm} \mist($G$)
\State {\bf Input:} A graph $G=(V,E)$.
\State {\bf Output:} A maximum cardinality of an independent set of $G$.
\If {$\exists v\in V$ with $d(v)=0$}
\State \Return {$1+\mist(G\setminus \SET{v})$}
\EndIf
\If {$\exists v\in V$ with $d(v)=1$}
\State \Return {$1+\mist(G\setminus N[v])$}
\EndIf
\If {$\Delta(G)\geq 3$}
\State {choose a vertex $v$ of maximum degree in $G$}
\State \Return {$\max(1+\mist(G\setminus N[v]), \mist(G\setminus \SET{v}))$}
\label{highestdegree}
\EndIf
\If {$\Delta(G)\leq 2$}
\State {compute $\alpha(G)$ using a polynomial time algorithm}
\State \Return {$\alpha(G)$}
\EndIf
\end{algorithmic}
\end{algorithm}

The algorithm is clearly correct since a vertex of degree~$0$ has no
neighbors in the current MIS being created and can be added to it. The same
applies to a vertex of degree~$1$, since there is no advantage to adding its
neighbor instead; its neighbor is discarded. If the degree is at least three,
one considers both possibilities, adding the vertex to the MIS and discarding
it. If all remaining vertices are of degree~2, the graph consists of
disjoint cycles, and it is easy to find maximal independent sets in cycles.

As a first intuitive explanation,
note that the algorithm gradually decreases the size of the graph until
the size of a maximal independent set is found, except that in
Line~\ref{highestdegree}, two options are explored recursively.
Using advice, one could simply make the correct choice of these two options.
A priority exact algorithm could be designed by trying all different
sequences of such choices.

In greater detail,
an input item is a vertex, together with a lists of
all its neighbors.
The history is known, so in designing priority functions, we can also talk
about the current degree, i.e., the number of neighbors that have not yet
been removed, as it is done in \mist.

In the priority exact algorithm we design the priority functions,
$P_i$, depending partially on the current degrees of the vertices.
Since neighbors of accepted vertices must be rejected,
these neighbors are given highest priority ($\Delta+3$, say).  Then,
vertices of current degree~$0$ have the next highest priority, $\Delta+2$,
vertices of current degree~$1$ have priority $\Delta+1$, and all other
vertices have priority equal to their current degree.

%At the end, we consider the computation of $\alpha(G)$ when the
%maximum current degree is~$2$, and there are no vertices of degree~$0$
%or~$1$, so there are only disjoint cycles remaining. At that point,
When there are only disjoint cycles remaining, we
define priority functions as follows: The lowest priority vertices are
those of degree~$2$, so they are not processed until it is time to
start a new cycle.  Every time we start the processing of a new cycle
(a degree~$2$ vertex), we accept the vertex (include it in the
maximum independent set).  The highest priority is given to vertices
adjacent to a vertex just processed. If it has current degree~$0$, it
is rejected, because it is adjacent to the first vertex in the
cycle. If it has current degree~$1$, it is accepted if its neighbor
was rejected and vice versa. Note that the priority does not alone
determine the decision made.

Advice comes into play in the case where the branching occurs, in
Line~\ref{highestdegree}. One bit of advice is used to tell which
branch gives the larger result, and the adaptive priority algorithm
with advice takes that branch, i.e., the advice is used to determine
if the vertex under consideration should be included into the maximum
independent set or not.  Note that the algorithm can easily determine
when to read a bit of advice, so the maximum amount of advice needed
is the number of branches on the shortest (meaning with fewest
branches) of the root to leaf paths that leads to a maximum
independent set. If one has a bound~$m$ on that number of branches in
the best case, it is never necessary to go through more than all $2^m$
possible bit strings of length $m$, and the natural approach is to do
the recursive branching with a bit in the advice string indicating
which branch to take.  In doing so, if one encounters an $(m+1)$st
branching, one can simply terminate computation in that direction and
move to the next bit string.  Thus, \mist can be seen as a priority
exact algorithm. Since the priority functions depend only on which
branches have been taken previously on the current root to leaf path,
it only depends on decisions made so far, so the defined priority
algorithm with advice is in the decision-based priority function
model.

The calculation of~$m$ is exterior to the algorithm and could, for
example, be an upper bound given as a function of~$\lvert V \rvert $.  By
recording accepted vertices, keeping the result with the
best~$\alpha(G)$, it is simple to return a maximum independent set
instead of just the size of it.

By the standard correspondence between Maximum Independent Set and
Minimum Vertex Cover, \mist can immediately be converted to an
algorithm for finding a minimum vertex cover\footnote{Minimum Vertex
  Cover is the problem of finding a minimum size subset of the
  vertices where every edge in the graph is incident to at least one
  of the vertices.} by reversing the decisions made.  In
Section~\ref{sec:vc}, \mist is extended to a priority algorithm with
advice, \VC, for finding minimum vertex covers in triangle-free graphs
of maximum degree~$3$, adding more priorities, particularly for
vertices of degree~$3$, considering which neighbors are shared with
previously processed vertices. \VC is shown to require at most
$7 \lvert V \rvert /22$ bits of advice, which is provably less advice
than required by any online algorithm with advice. No adaptive
priority algorithm without advice can achieve an approximation ratio
for this problem better than~$4/3$~\cite{BorodinBLM2010}. Thus, \VC is
evidence that the class of adaptive priority algorithms with advice is
a larger class than either of these related classes of algorithms.

Running the algorithm \VC on all possible advice strings of
length~$7n/22$, we obtain an offline algorithm solving the problem to
optimality, a priority exact algorithm, that runs in time\footnote{The
  notation $\Oh{}$ is similar to big-Oh, except that it allows
  ignoring polynomial factors, i.e., $\Oh{g(n)}$ has the same meaning
  as $O(g(n) \poly(n))$.}
$\lOh{2^{\frac{7 n}{22}}}\subset \Oh{1.247^n}$. This is much better
than the naive $\Oh{2^n}$ brute-force approach; however, there are
other more involved optimal offline algorithms achieving even better
runtimes for the Minimum Vertex Cover problem.  The best published
exact algorithm for Minimum Vertex Cover restricted to graphs of
maximum degree~$3$ runs in
$\Oh{1.0836^n}$~\cite{DBLP:journals/tcs/XiaoN13}. That algorithm is
not a priority exact algorithm; in Section~\ref{sec:firstexample} and
Subsection~\ref{sec:lbopt}, we show that no priority exact algorithm
(derived from a priority algorithm with advice in the decision-based
or oblivious priority function models) for Minimum Vertex Cover on
triangle-free graphs of maximum degree~3 has a running time less than
$\Omega(1.142^n)$.  We comment further on this in
Subsection~\ref{sec:lbopt}.

%In \cite{AHI05}, lower bounds for solving SAT are given for
%generalized myopic algorithms and drunk algorithms.  

%two families of DPLL algorithms (Davis,
%Putnam, Logemann, Loveland~\cite{DLL62,DP60}). 

%DPLL algorithms are branch-and-reduce straties,
%based on assigning values to the variables. %The drunk algorithms resemble extensions of
%priority algorithms, in that on can choose the next variable using any rule and then pick its
%value at random. They use treelike resolution to prove for drunk DPLL algorithms.
%Other lower bounds for DPLL algorithms for $k$-SAT include~\cite{}.
  %ABM04}.
%
%Exponential lower bounds for other classes of (what can be seen as) branch-and-reduce algorithms
%for
%$k$-SAT ($2^{n(1-c_k)}$, for some $c_k\in O(1/k^{1/8})$ almost general DPLL algorithms)~\cite{PI00}
%+ AHI05,
%Maximum Independent Set ($c^n$ for some $c>1$ for ``recursive proofs'')~\cite{C77},
%Graph Coloring ($c^{n(ln n)^{1/2}}$ for some $c>1$ for ``Zykov algorithms'')~\cite{M79},
%Knapsack ($2^{n/10}$ for ``recursive algorithms'')~\cite{C80}

\subsection{Priority Exact Algorithms, in General}

When attacking new NP-hard problems, the priority exact
algorithms approach has the potential to deliver a first upper bound
that beats the brute force approach, giving an aim for later, more
specialized, possible improvements.

A significant motivation for originally introducing and studying
priority algorithms was to develop a framework for proving lower
bounds for a large collection of algorithms at the same time:
Establishing that no fixed (or adaptive) priority algorithm can attain
a certain approximation ratio implies that one has to look beyond this
fairly broad design pattern to possibly discover an algorithm with a
better approximation ratio. We note that this motivation is just as
relevant for the design of exact or approximation algorithms using the
framework outlined above. A discussion of the lower bound results we
obtain is included in Subsection~\ref{sec:lbopt}.

Priority exact algorithms form a subset of the more general
branch-and-reduce~\cite{DP60,DLL62} exact algorithms, which find an optimal
solution to a problem using a search tree and backtracking. Trying
successive possibilities for the
advice, setting some decision to accept or reject for example, is
essentially the same as a branch operation in the more general
algorithms. %, except for the restriction to only one input item (a
%vertex in the case of many graph problems).
The restriction that
input items be prioritized independently of each other means that
there are many possibilities allowed in the general branch-and-reduce
algorithms that
are not allowed in priority exact algorithms. For the Minimum Vertex
Cover problem, for example, priority exact algorithms cannot handle
maximal connected
components of size at most~$20$ separately (or even handle a vertex of
degree~$2$ differently depending on whether or not it is contained in
a triangle); in fact, the lower bounds are proven by considering small
connected components.

While there are restrictions, the advantage of priority exact algorithms is
that they should be relatively easy to implement and efficient (other
than the branching, of course).
A straight-forward implementation of a priority exact algorithm as
a branching algorithm may lead to many fewer branches than one would
obtain by enumerating all bit strings of the maximum length, even in the
worst case. In many cases the problem size would reduce by
different amounts, depending on whether the decision was accept or
reject, for example.  One could also apply standard techniques for
establishing upper bound results, such as
measure and conquer~\cite{DBLP:journals/jacm/FominGK09} to obtain better upper
bounds.

In general, branch-and-reduce algorithms can be considered to have been
converted from (usually not priority) algorithms with advice. Advice
can be given for each node in the search tree indicating which branch
to take to find an optimal solution.  If the work done at a node
can be handled by a priority algorithm (and all root to leaf paths
have the same length), then it is essentially a priority exact
algorithm. However, for example for Minimum Vertex Cover, most exact
algorithms use operations that do not fit in the priority algorithm
model.

A lower bound, related to those presented here for priority exact
algorithms, is presented in~\cite{C77}, where the lower bound also holds
for priority exact algorithms (recursive proofs) for Maximum
Independent Set (if one ignores cutting off the length of the root to
leaf paths considered due to the maximum length of the advice string
necessary), but also for more powerful algorithms, and proves that
there exists a $c>1$ such that the running time is at
least~$\Omega(c^n)$. In fact, this result holds for every graph in a
large class.

Exponential lower bounds for other classes of (what can be seen as)
branch-and-reduce algorithms exist for other problems as well, for example
$k$-SAT~\cite{PI00,ABM04,AHI05}, Maximum Independent Set~\cite{C77},
Graph Coloring~\cite{M79}, and Knapsack~\cite{C80}.

\section{Example: Minimum Vertex Cover}
\label{sec:firstexample}
We now present an example, mainly illustrating some of our techniques for
proving lower
bounds for priority algorithms with advice, but also presenting an
algorithm showing that the result is tight for the class of inputs
given by the adversary. Both the algorithm and lower bound apply to the
decision-based priority model. 

Given a simple undirected graph $G = (V,E)$, a subset of vertices $S
\subseteq V$ is called a \emph{vertex cover} if every edge is incident
to at least one vertex from $S$. \emph{Minimum Vertex Cover} is the
problem of finding a vertex cover of minimum size.
An input item is a vertex together with a complete list of its
neighbors (including those vertices that have not even appeared as
part of the input yet); this is known as the \emph{vertex arrival,
  vertex adjacency model}.  Thus, for each vertex, when it becomes the
highest priority vertex, the priority algorithm must decide whether or
not to ``accept'' or ``reject'' that vertex, under the condition that
at the end, for every edge in the graph, at least one of its endpoints
must have been accepted.

\begin{theorem}
  \label{MVClower}
No adaptive priority algorithm can solve Minimum Vertex Cover optimally with fewer than $|V|/7$ bits of advice in the decision-based priority function model.
\end{theorem}
\begin{proof}
  Within the proof, we have found it beneficial to include intuition and introduce terminology relevant for the general templates, making the style somewhat different from a normal formal proof.
 
We build on the construction in~\cite{BorodinBLM2010}
(which was reused in~\cite{BBLP20}), showing that for this problem, no
adaptive priority algorithm without advice can achieve an
approximation ratio better than~$4/3$. The two graphs in
Fig.~\ref{Graph3} are used.

\begin{figure}[htp]
\centering
\begin{tikzpicture}[scale=0.5]

\node[draw=black,circle,fill=white] (1) at (0,-8) {4};
\node[draw=black,circle,fill=lightgray] (2) at (-3,-6) {3};
\node[draw=black,circle,fill=white] (3) at (0,-6) {7};
\node[draw=black,circle,fill=lightgray] (4) at (3,-6) {5};
\node[draw=black,circle,fill=white] (5) at (-3,-4) {2};
\node[draw=black,circle,fill=white] (6) at (3,-4) {6};
\node[draw=black,circle,fill=lightgray] (7) at (0,-2) {1};

\draw (1) -- (2);
\draw (1) -- (4);
\draw (2) -- (3);
\draw (3) -- (4);
\draw (2) -- (5);
\draw (5) -- (7);
\draw (4) -- (6);
\draw (6) -- (7);
\end{tikzpicture}
\hspace{1cm}
\begin{tikzpicture}[scale=0.5]

\node[draw=black,circle,fill=lightgray] (1) at (0,-2) {4};
\node[draw=black,circle,fill=lightgray] (2) at (0,-5) {3};
\node[draw=black,circle,fill=white] (3) at (3,-2) {7};
\node[draw=black,circle,fill=white] (4) at (3,-8) {5};
\node[draw=black,circle,fill=lightgray] (5) at (0,-8) {2};
\node[draw=black,circle,fill=white] (6) at (3,-5) {6};
\node[draw=black,circle,fill=white] (7) at (-3,-5) {1};

\draw (1) -- (7);
\draw (5) -- (7);
\draw (2) -- (7);
\draw (1) -- (3);
\draw (2) -- (3);
\draw (2) -- (6);
\draw (5) -- (6);
\draw (5) -- (4);
\draw (4) to[out=45,in=45,distance=55mm] (1);
\end{tikzpicture}
\caption{Topological structures of graphs giving a lower bound for the
  Minimum Vertex Cover problem. Graph~1 is on the left and Graph~2 is on the
  right. The unique minimum vertex covers are marked in gray.}
\label{Graph3}
\end{figure}

In proving lower bounds for adaptive priority algorithms, the
adversary chooses the input, first choosing the universe of input
items, and then creating an actual input $X$ from that
universe. Originally the adversary can set $X$ to the entire
universe. Then it (perhaps gradually) removes input items from $X$ as
the algorithm selects input items using priority functions and makes
irrevocable decisions for them. Thus, the input item selected by the
current priority function is always one of the remaining input items in
$X$ with highest priority. When there are ties, the adversary can
choose among those with highest priority.  (In this proof, the
adversary can simply choose an arbitrary item of highest priority, so
we may assume that there is always a single input item with highest
priority.)

For Minimum Vertex Cover, the adversary, \ADV, will select an
isomorphic copy of either Graph~1 or Graph~2 from Fig.~\ref{Graph3},
depending on the algorithm, \ALG.  Since both graphs have seven
vertices, the universe, $\mathcal{U}$, of input items, contains the
names of seven vertices (the same names are used for both graphs), and
for each of the vertices, all possibilities for input items (names of
vertices and lists of neighbors) for degrees two and three.  Note that
both graphs have unique minimum vertex covers of size~$3$. The numbers
shown in the figure are for our reference only and do not represent
actual input items given to an algorithm. The figure represents the
\emph{topological} structure of the inputs. The actual input items
would be created out of all consistent namings of vertices in such
graphs. To illustrate this point, consider vertex~$1$ in Graph~$1$. It
is adjacent to vertices~$2$ and~$6$. The corresponding possible input
item could happen to be $(1, \{2, 6\})$, but it could also be $(5,
\{2, 3\})$, for example. In the latter case, the actual input
vertex~$5$ would be mapped to vertex~$1$ in the figure, vertex~$2$
would be mapped to vertex~$2$, and vertex~$3$ would be mapped to
vertex~$6$. In total, there are $7\times 6 \times 5$ possible input
items that could be associated with vertex~$1$ in Graph~$1$. After a
particular item has been processed, the number of items that could be
associated with subsequent vertices is reduced because of consistency
requirements.

Given this universe of input items, the first priority function for
any algorithm, \ALG, for Minimum Vertex Cover must select either a
vertex of degree~$2$ or a vertex of degree~$3$ as the first vertex to
be processed.

In order to obtain a vertex cover of size~$3$, it is necessary to
accept vertex~$1$ in Graph~1 and reject vertex~$2$ in Graph~1. Thus,
for the case where the first vertex selected by \ALG
has degree~$2$, \ADV can force \ALG to produce a vertex cover of size
at least four by choosing vertex~$1$ from Graph~1 if \ALG rejects
and choosing vertex~$2$ from Graph~1 if \ALG accepts. Because of
how the universe is defined, \ADV can do this regardless of which
input item with degree~$2$ \ALG chooses.

Similarly, in order to obtain a vertex cover of size~$3$, it is
necessary to accept vertex~$3$ in Graph~1 and reject vertex~$1$ in
Graph~2. Thus, for the case where the first vertex selected by \ALG
has degree~$3$, \ADV can force \ALG to produce a vertex cover of size
at least four by choosing vertex~$3$ from Graph~1 if \ALG rejects and
choosing vertex~$1$ from Graph~2 if \ALG accepts. Again, because of
how the universe is defined, \ADV can do this regardless of which
input item with degree~$3$ \ALG chooses.

To define a problem where $k=|V|/7$ bits of advice are necessary
for optimality in the decision-based priority model, we
consider an algorithm, \ALGp, and an adversary, \ADVp.  We create $k$
disjoint subuniverses,
$\mathcal{U}_1,\mathcal{U}_2,\ldots,\mathcal{U}_k$, copies of the
subuniverse~$\mathcal{U}$, with different names for the vertices in
each copy, and define the universe, $\mathcal{U}'$, for \ALGp to be
the union of these $k$ subuniverses.  The input for \ALGp is the union
of $H_1,H_2,\ldots,H_k$, where $H_i$ is an isomorphic copy of either
Graph~1 or Graph~2.

With its priority functions, \ALGp can choose input items in many
different ways, and could, for instance, interleave input items
stemming from different copies of~$\mathcal{U}$.  However, for each $\mathcal{U}_i$, there
is always a first vertex in $\mathcal{U}_i$ that \ALGp chooses (from the current
subset $X$ of the universe, $\mathcal{U}'$).  When \ADVp is not restricted by
advice that \ALGp has read, it can force \ALGp to accept a vertex
cover of size four for $H_i$, exactly as \ADV forces \ALG, depending
on whether this first vertex from $\mathcal{U}_i$ has degree~2 or~3.

We now define $2^k$ sequences of input items for \ALGp, by describing
how one of these $2^k$ sequences of input items is defined: \ALGp
selects input items one at a time, and \ADVp knows from which of the
$k$ subuniverses the input items originate.

In this concrete case of an adaptive priority algorithm (with advice),
since we are assuming that \ALGp solves the problem to optimality, the
adversary can assume in the decision-based priority model that the
current priority function is determined based on \ALGp making the
correct accept/reject decisions up to this point. Now, \ADVp does the
following: Assume that \ALGp has already received input items
originating from $i$ of the subuniverses from which $\mathcal{U}'$ was
defined and the adversary has a current subset $X\subseteq U'$.  If
that is the case, then $X$ contains exactly enough input items to
complete one graph from each of the subuniverses from which \ALGp has
received some input item (how this is maintained is explained below).
From subuniverses not included in these $i$ subuniverses, $X$ still
contains all possible names for vertices in the graphs.

Now, \ALGp receives its next input item which will be the input item in $X$
of the highest priority in this round, and that input item is the next
in the input sequence we are defining. This item is determined by the
current priority function which only depends on the input items
received so far and its decisions so far.

If that next input item, $v$, is from one of the $i$ subuniverses,
nothing further is done. However, if that next input item originates
from a subuniverse not among the $i$, then the following is done.

If $v$ has degree~2, \ADVp can choose that it is vertex~$1$ in Graph~1
or vertex~$2$ in Graph~1. If $v$ has degree~3, \ADVp can choose that
it is vertex~$3$ in Graph~1 or vertex~$1$ in Graph~2.
It makes a choice and then removes from $X$ all input items originating from
the selected subuniverse of $\mathcal{U}'$, except enough to make up
exactly the graph that was chosen (Graph~1 or Graph~2) with the
vertex names consistent with the first input item from that graph.

Continuing this inductively defines one of the $2^k$ distinct input
sequences.

If a priority algorithm with advice for Minimum Vertex Cover uses
fewer than $k$ bits of advice for instances with $7k$ input items, the
same advice must be given for at least two of the sequences, $I_1$
and~$I_2$, defined above. \ALGp therefore uses the same priorities and
makes the same decisions on $I_1$ and $I_2$ until some difference is
detected.  Thus, consider the first time in the processing of $I_1$
and~$I_2$, where an input item that has current highest priority is
the first input item of a graph from some~$\mathcal{U}_j$, but the graphs
included in $I_1$ and~$I_2$ from~$\mathcal{U}_j$ are different.

Up until (and including) this point, all input items have been the
same for the two sets. Thus, \ALGp must make the same decision for~$v$
in both $I_1$ and~$I_2$, but one of those decisions leads to a vertex
cover of size four. Thus, \ALGp is not optimal, and $k$~bits of advice
are necessary.
\end{proof}

This lower bound is generalized in Subsection~\ref{sec:lbopt}, giving
a template for proving such bounds.

For an algorithm matching the lower bound of the above theorem
on these particular types of inputs using Graphs~1~and~2, we
begin with the case~$k=1$, i.e., we receive a graph isomorphic
to either Graph~1 or~2.
Making the correct decision on the first vertex received
enables a priority algorithm to obtain a vertex cover of size~$3$ by
giving highest priority after that to neighbors of vertices which are
already chosen, accepting if the known neighbor was rejected, and
rejecting if the known neighbor was accepted. Continuing in this way
until all vertices are processed always produces the minimum vertex cover.
Thus, one bit of advice is necessary and sufficient
for optimality for these restricted inputs; the one bit indicates
whether or not the first vertex should be accepted or rejected.

Extending the algorithm just described for the case $k=1$ for
achieving optimality when one bit of advice is given per subuniverse,
one notes that $k$~bits of advice are also sufficient for these
very specific types of input. Thus, in this very restricted
problem, for every positive integer $k$, there is an input size
where $k$ bits of advice are necessary and sufficient.

Since the results in this section concern exact, rather
than approximation algorithms, all results also apply to Maximum
Independent Set for graphs of maximum degree~$3$. Both Graph~1
and Graph~2 are triangle-free graphs of maximum degree~3, so
the lower bound also holds for triangle-free graphs of maximum degree~3,
as does the $4/3$ lower bound on the approximation ratio for
adaptive priority algorithms without advice.

\section{Solving Minimum Vertex Cover to Optimality for Triangle-Free Graphs of Maximum Degree~$\mathbf{3}$}
\label{sec:vc}

We consider the Minimum Vertex Cover problem, as defined in
Section~\ref{sec:firstexample}, on triangle-free graphs of maximum
degree~$3$, in the online and in a priority setting with advice. The
vertex arrival, vertex adjacency model is used.  (Since the
results in this section concern exact, rather than approximation
algorithms, all results also apply to Maximum Independent Set for
triangle-free graphs of maximum degree~$3$.)  Let $n$ denote the
number of vertices in the input graph. As mentioned in
Section~\ref{sec:firstexample}, no adaptive priority algorithm without
advice can achieve an approximation ratio for this problem better than
$4/3$~\cite{BorodinBLM2010}, since graphs used in the construction
there were triangle-free with maximum degree~$3$. In this section, we
show that asymptotically this problem requires at least $(n-4)/3$ bits
of advice to solve optimally in the online setting, while it can be
solved optimally using at most $7n/22 < 0.3182n$ bits of advice in the
adaptive priority setting.

We begin with the negative result for the online setting.
\begin{theorem}
\label{thm:online_vc_lb}
Asymptotically, for $n\geq 7$, no online algorithm using fewer than
$(n-4)/3$ bits of advice can accept a minimum-sized vertex cover for
all triangle-free graphs of maximum degree~$3$.
\end{theorem}
\begin{proof}
The adversary will use a graph with $n= 6n'+1$ vertices, where $n'\geq
2$. The set of all vertices is denoted by $V$.

One way to describe the adversarial input is as if it is being
constructed in stages. In the first stage, the adversary creates $2n'$
disconnected paths of length $2$ each, or $2$-paths, for short (this
already gives $6n'$ nodes). In the second stage, the adversary
connects endpoints of $2$-paths, chaining
several paths together into one large cycle.
Not all initial $2$-paths will necessarily participate in the
cycle. Finally, the adversary attaches one more vertex to an
appropriately chosen vertex~$v$ in the cycle and decides how to
present this constructed graph online. An optimal decision to accept
or reject a middle vertex of each initial $2$-path depends on the
answers to these questions: Does this $2$-path participate in the
large cycle or not, and, if it participates in the cycle, is it
located at an even or odd distance from~$v$. When the fully
constructed adversarial input is presented to an online algorithm such
that middle vertices of initial $2$-paths are given first, the
algorithm does not yet know the answers to the questions above, so a
lot of advice is required to infer correct decisions for these
vertices.

More formally, let $S=\{v_1,v_2,\ldots,v_{2n'}\}$ be the first $2n'$
vertices to be given -- they form middle vertices of $2$-paths, so all
vertices in $S$ will have degree~$2$. Throughout the processing of
$S$, the neighbors will be vertices never seen before.  As described
above, some neighbors of $S$ will be connected so as to form a cycle,
which we denote by~$C$.  Then there will be a unique vertex~$w$ of
degree~$1$, connected to one designated neighbor $v\in C\setminus S$.
Finally, the set of all other vertices will be denoted $I$, i.e., $I =
V \setminus (C \cup \{w\})$. This set induces isolated $2$-paths, with
the middle vertices in $S$.  The vertex~$v$ will have
degree~$3$. There will be an even number of vertices from $S$ in $I$
and, thus, an even number in $C$. The construction is illustrated in
Fig.~\ref{fig:online_vc_lb}.

\begin{figure}[htp]
\centering
\includegraphics[width=\linewidth]{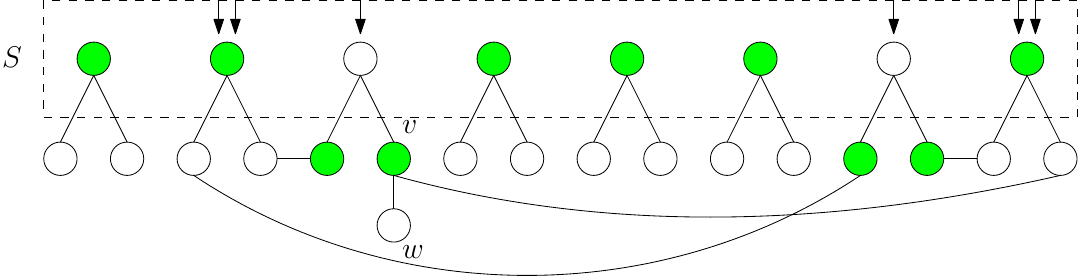}
\caption{The
    construction used in Theorem~\ref{thm:online_vc_lb}. Here, we have
    $n' = 4$. The optimal vertex cover is shown in green. Vertices with a
    single arrow pointing to them are those vertices from $S$ that were
    selected to be at odd distance from node $v$. Vertices with two
    arrows pointing to them are those nodes from $S$ that were
    selected to be at even distance from node $v$. Here,
    the number of vertices in $S$ \emph{not} in the optimal vertex
    cover is $r=2$.}
  \label{fig:online_vc_lb}
\end{figure}

Note that this graph has a unique
minimum-size vertex cover: the middle vertex of each path in $I$ and
every other vertex in $C$, starting with~$v$.

For each vertex, $u\in S$, all of which have degree~$2$, $\ALG$ must
decide whether to accept or reject this vertex, without knowing if $u$
is in $I$ or $C$. Of course, within $C$, $\ALG$ will not know if
$u$ will be at an even or odd distance from~$v$.

Suppose we want to create a graph $G$ with $0\leq r \leq n'$ vertices
from $S$ \emph{not} in the optimal vertex cover. We can choose any
subset $R$ of $r$ vertices in $S$ to be at odd distances from $v$ in
$C$. Among the other vertices, $r$ can be placed at the even locations
in $C$, and the remaining $2n'-2r$ vertices from $S$ can be in
$I$. (The placement of $v$ is also arbitrary, but we are fixing a
placement in this counting.)  For fixed $r$, there are
$\binom{2n'}{r}$ different possibilities for the subset $R$. In all,
there are $\sum_{r=0}^{n'} \binom{2n'}{r}$ different possibilities for
the subset $R$, each with a different optimal vertex cover (note that
$r=0$ is a degenerate case where there is no cycle, $v$, or $w$, but
the instance is still a possibility, and for $r\geq 1$, the unique
cycle $C$ has at least $6$ vertices and $n\geq 7$, so the graph is
triangle-free).  Any online algorithm that gets the same advice for
two of them must give a suboptimal cover for at least one of
them. Thus, an algorithm that solves the problem to optimality needs
at least
$\log_2 \sum_{r=0}^{n'} \binom{2n'}{r} >\log_2 2^{2n'-1}
=2n'-1=(n-4)/3$ bits of advice.

Just for emphasis, note that all input items in $S$ are fixed to be
exactly the same in all instances that we consider, i.e., input items
in $S$ do not depend on the choice of $R$, $v$, and $w$. Thus, an
online algorithm receiving items from $S$ can only rely on advice to
act differently on $S$ from instance to instance.
\end{proof}

Now, we present an \emph{adaptive priority algorithm with advice} that
works in both the decision-based and oblivious priority function
models, uses fewer than $(n-4)/3$ bits of advice, and achieves
optimality.

We present an adaptive priority algorithm \VC with advice for
the Minimum Vertex Cover problem on triangle-free graphs of maximum
degree~$3$. The main result of this section is the following:

\begin{theorem}
  \label{thm:optimal-algorithm}
  \VC solves Minimum Vertex Cover on triangle-free graphs with maximum
  degree~$3$ optimally in both the decision-based and oblivious
  priority function models and uses at most
  $(7/22)n = 0.31\overline{81} n$ bits of advice, where $n$ is the
  number of vertices.
\end{theorem}
\begin{proof}
Follows from Lemmas~\ref{lem:vc-corr} and~\ref{lem:vc-adv-len}.
\end{proof}

In order to describe and analyze the algorithm, we have to introduce
and define some terminology. We do this in the order from most
intuitive to least intuitive. Fortunately, most of the terminology
will be self-explanatory, but needs to be stated for the sake of
completeness.

Since it is an adaptive priority algorithm, \VC works in
discrete time steps. Each time step consists of the algorithm updating
the priority function, receiving the next input item according to the
new priority, potentially reading advice, and then making a decision
as to including the vertex corresponding to the input item in
the solution or not. We also refer to the decision of including the
vertex in the solution as \emph{accepting the vertex} and the opposite
decision as \emph{rejecting the vertex}.  The decision is called
\emph{correct} if it is possible to extend the partial solution
obtained after the decision to a minimum vertex cover in the input
graph.

In many cases it is possible to make a decision that is guaranteed to
be correct without consulting advice at all. Consider, for example, a
vertex of degree~$1$ -- it is easy to see that a correct decision is
to reject such a vertex and then accept its unique neighbor.

Suppose that at time $t$ vertex~$v$ arrives and it is not possible,
from the vertices seen so far, to make a decision that can be
guaranteed to be correct no matter what happens in the rest of the
input. In this case, \VC reads a single bit of advice. This
bit encodes a correct decision for the algorithm. In other words, if
the bit is $1$, then the algorithm accepts $v$ and otherwise the
algorithm rejects~$v$. In these cases, we say that the advice is to
accept or reject the vertex, respectively. We also say that $v$
received advice.

Once a decision has been made for a vertex, this vertex is called
\emph{processed}. Vertices that have not been processed are called
\emph{unprocessed}. Suppose that the algorithm processes the vertices
in the order $v_1, v_2, \ldots, v_n$ -- this notation is only for the
duration of this paragraph and will have a different meaning in the
proofs below. Recall that input items corresponding to the vertices
consist of pairs $(v_i, N(v_i))$, where $N(v_i)$ is the neighborhood
of the vertex~$v_i$. Since the priority algorithm is adaptive, it can
effectively remove processed vertices from the input graph. Namely, at
time $i$, the algorithm knows~$v_1, \ldots, v_{i-1}$. Therefore, in
defining the priority function, the algorithm can ignore vertices in
$\{v_1, \ldots, v_{i-1}\}$ when assigning a priority to $(v, N(v))$,
which is equivalent to removing vertices $v_1, \ldots, v_{i-1}$ from
the rest of the input graph. We refer to
$N(v)\setminus\{v_1, \ldots, v_{i-1}\}$ as the \emph{current
  neighborhood of $v$} and~$|N(v)\setminus\{v_1, \ldots, v_{i-1}\}|$
as the \emph{current degree of $v$}.
We refer to $N(v)$ and~$|N(v)|$ as the \emph{original
  neighborhood} of $v$ and the \emph{original degree} of $v$,
respectively.

The following is less intuitive but useful terminology for vertices:
\begin{description}
\item[aa-vertex:] a processed vertex that received advice to be
  accepted.
\item[ar-vertex:] a processed vertex that received advice to be
  rejected.
\item[a-vertex:] either an aa-vertex or an ar-vertex.
\item[non-a-vertex:] a vertex that was processed without advice.
\item[contributing:] an aa-vertex that has two rejected and one
  unprocessed neighbor.
\item[c-neighbor:] an unprocessed vertex that is a neighbor of a
  contributing vertex.
\item[bad-vertex:] a vertex that requires advice and all of its
  neighbors are c-neighbors of other vertices at the time this vertex
  is processed.
\item[a-sibling:] a neighbor of an aa-vertex~$v$ such that $v$ has
  another neighbor that has been
  accepted.
\end{description}
Observe that the above definitions are with respect to a given time
step. In particular, it is possible that a vertex~$v$ is processed
during some time step and at that point becomes an aa-vertex. At a
later time step, it could become a contributing vertex.
Also observe that it is possible that a neighbor of an
unprocessed vertex is a c-neighbor, that is, a neighbor of some other
vertex that is contributing at the time of consideration.

The pseudocode of \VC is given in
Algorithm~\ref{alg:rejectfirst}. Ties that are not broken by
\VC explicitly can be broken arbitrarily (even by an adaptive
adversary).

\newlength{\myindent}
\settowidth{\myindent}{P0: \mbox{}}
\newcommand{\s}{\hspace*{\myindent}}

\begin{algorithm}[!h]
\caption{\VC algorithm.}\label{alg:rejectfirst}
\begin{algorithmic}
\Procedure{\VC}{}
\While{there exist unprocessed vertices} 
    \State{Define the priority function $P$ as follows}
    \State{\s (listed in order from highest to lowest priorities):}
    \State{P1: nodes with a rejected neighbor;}
    \State{\s highest priority is given to those nodes whose neighbor was most}
    \State{\s recently rejected.}
    \State{P2: nodes with current degree~$0$.}
    \State{P3: nodes with current degree~$1$;}
    \State{\s highest priority is given to those nodes with a most recently}
    \State{\s processed neighbor; among those, highest priority is given to}
    \State{\s those nodes that had two neighbors that became aa-vertices.}
    \State{P4: nodes with current degree~$2$ that had a third neighbor in common}
    \State{\s with a previously rejected bad-vertex.}
    \State{P5: a-siblings.}
    \State{P6: nodes with current degree~$3$ with $2$ or $3$ neighbors in common}
    \State{\s with a single aa-vertex that was not a bad-vertex when it received}
    \State{\s advice.}
    \State{P7: nodes with current degree~$3$ that share neighbors with a-vertices.}
    \State{P8: other nodes with current degree~$3$.}
    \State{P9: nodes with current degree~$2$}

    \State{\mbox{}}

    \State{Receive the next vertex~$v$ according to $P$}

    \Switch{priority of $v$}
        \Case{P1 or P6:}
            \State{Accept $v$}
        \EndCase
        
        \Case{P2, P3, P4, P5, or P9:}
            \State{Reject $v$}
        \EndCase
        
        \Case{P7 or P8:}
            \State{Obtain advice to accept or reject and apply it to $v$}
        \EndCase
    \EndSwitch
   
\EndWhile
\EndProcedure
\end{algorithmic}
\end{algorithm}

In order to finish the specification of \VC, we have to
describe how the oracle generates the advice. The oracle sees the
entire input beforehand and it knows how \VC works. Since
\VC is deterministic, the oracle can, in effect, simulate
\VC on the input. Thus, the oracle knows the order in which
the vertices are processed and it knows at which time steps
\VC asks for advice. The oracle supplies the advice in the
order in which the advice is requested by \VC. Suppose that
at some time, \VC processes $v$ based on advice. If there
is a unique correct decision for $v$, the oracle provides that
decision, either accept or reject, which is one bit of information. If
either decision is correct (could be completed to a minimum vertex
cover) and $v$ is a bad-vertex, the oracle advises to accept. Finally,
if either decision is correct and $v$ is not a bad-vertex, the oracle
advises to reject. This tie-breaking condition is particularly
important for the analysis.

We mention a few high level features of \VC. Vertices that
obviously can be handled without advice are those with current degree
$0$ or~$1$, and neighbors of rejected (accepted in \mist) vertices.
The two key observations in the design of \mist are the following:
First, the vertices just
described should receive the highest priorities (as described in
Section~\ref{misthree}).
Second, if we
process vertices of current degree~$3$ prior to processing vertices of
current degree~$2$ (with a small exception of P4; ignore for the
moment), then, when a vertex of current degree~$2$ arrives according
to~P9, we know that all the remaining vertices in the graph have
current degree~$2$. We can conclude that the remaining graph is a
collection of disjoint cycles and an optimal vertex cover in such a
graph can be computed by a priority algorithm without
advice. Therefore, with such an approach, only vertices of current
degree~$3$ may require advice and the goal is to minimize the
number of such vertices. This is where cooperation between the oracle
and the algorithm becomes crucial -- we shall see that the
tie-breaking condition of the oracle is chosen so as to create
scenarios under which some vertices of current degree~$3$ may be
processed without advice.

Next, we analyze \VC formally. We begin with the easier proof
of correctness of the algorithm and then establish the sufficient
number of bits of advice. Suppose that at time $t$ a vertex arrives
according to priority P9 for the first time. Then we refer to the time
interval $[1, t-1]$ as \emph{Phase~$1$} and to the time interval
$[t, n]$ as \emph{Phase~$2$}. If such $t$ does not exist then we set
$t=n+1$ meaning that the entire time interval $[1, n]$ consists of
only Phase~$1$ and Phase~$2$ is empty. Correctness of the algorithm
follows from the following lemma.

\begin{lemma}
\label{lem:vc-corr}
Every decision of \VC is correct.
\end{lemma}
\begin{proof}
  The proof is by simple induction: if all previous decisions are
  correct, we need to demonstrate that the decision for the next
  vertex is also correct. We omit the formal setup of induction and go
  straight to the inductive step. Let $v$ be the newly arriving
  vertex.

First, suppose that the algorithm is in Phase~$1$. 

\textit{Case: $v$ has priority P1, P2, P3, P7, or P8.} The decisions
of \VC are obviously correct.

\textit{Case: $v$ has priority P4.} \VC rejects $v$, so
suppose for the sake of contradiction that $v$ should have been
accepted. Let $v'$ denote a bad-vertex and~$u_1$, $u_2$, and $u_3$ its
three neighbors such that $u_1$ is also a neighbor of~$v$. This is
illustrated in the picture below, which omits some edges so as to
avoid clutter.

\begin{center}
\includegraphics[scale=0.6]{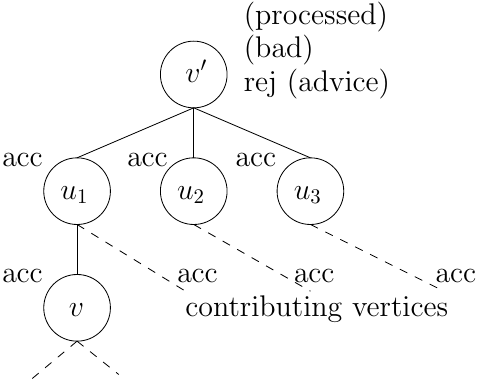}
\end{center}

Observe that another optimal vertex cover is obtained by accepting
$v', u_2, u_3, v$ and rejecting~$u_1$. Thus, the oracle would have
given advice to accept $v'$, since at the time $v'$ was processed,
both decisions were correct, and the oracle prefers accepting
bad-vertices. This is a contradiction, so the decision of \VC
to reject $v$ is correct.

\textit{Case: $v$ has priority P5.} Observe that processing
bad-vertices leads to processing of their neighbors prior to any
vertex with priority P5 being processed. Therefore, $v$ is a neighbor
of an aa-vertex~$v'$ and~$v'$ was never bad. Denote the neighbors of
$v'$ by $u_1, u_2, u_3$ such that $u_1 = v$. Consider the time when
$v'$ received advice to be accepted. We claim that at most one of
$u_1, u_2, u_3$ can be accepted in the future. Suppose, for the sake
of contradiction, that at least two nodes, say, $u_1$ and $u_2$, must
be accepted in the future. Then accepting $u_1, u_2, u_3$ and
rejecting $v'$ would result in a vertex cover of the same size or
smaller as accepting $v', u_1, u_2$ and rejecting/accepting~$u_3$. In
this case, since $v'$ was not bad at the time it received advice, the
oracle should have given advice to reject $v'$ according to the
tie-breaking condition. This is a contradiction, and therefore at most
one of $u_1, u_2, u_3$ can be ever accepted. By definition of an
a-sibling, either $u_2$ or $u_3$ has been accepted prior to $v=u_1$
being processed, so it is correct to reject~$v$.

\textit{Case: $v$ has priority P6.} As in the previous case, let $v'$
be the aa-vertex that shares at least two neighbors with $v$ and that
was not bad at the time it was processed. As already argued, at most
one neighbor of $v'$ can be accepted, therefore at least one of the
neighbors of $v$ in common with $v'$ must be rejected. Since each edge
must be covered by the solution, we conclude that $v$ must be
accepted.

Since the case of~P9 cannot happen in Phase~1, we move to the analysis
of Phase~$2$. As discussed prior to this lemma, at the beginning of
Phase~2 we know that the remaining graph is a collection of
cycles. Once a vertex of current degree~$2$ arrives according to~P9,
it is rejected, which creates two vertices of current degree~$1$
each. They are neighbors of a rejected vertex, so they are processed
next according to~P1. The degrees of their neighbors on the cycle drop
to $1$ or $0$, so they are processed according to~P1--3. This
continues until all vertices in this cycle have been processed. Then
the next cycle is processed and so on. The correctness of the
constructed vertex cover follows from the fact that a minimum vertex
cover in every cycle rejects at least one vertex. By symmetry, a
minimum vertex cover may be rotated clock-wise so any vertex may be
that rejected vertex. Thus, it is always safe to reject the first
vertex from the cycle. After that, correctness follows by the
correctness of cases~P1--3, as in Phase~$1$.
\end{proof}

Central to the analysis of the number of bits of advice is the notion
of a \emph{component}. A new component starts when a new a-vertex is
processed that does not have neighbors in common with a previously
processed a-vertex.  When a new component is started, any previous
component is closed, meaning that it receives no more vertices. A
vertex is included in the current component if it is not in any
previous component, and one of the following cases applies:
\begin{itemize}
\item it is an a-vertex that shares a neighbor with a previously
  processed a-vertex from the current component,
\item
it is a neighbor of an a-vertex from the current component,
\item
  it is accepted or rejected before the component is closed.
\end{itemize}

Note that a component in the above sense is not to be confused with a
connected component -- it is possible for a connected graph to consist
of several components, and it is possible that such a component is not
connected.

We let $c$ denote the final number of components created by
\VC on the given input. For $i \in [c]$, we let $a_i(t)$
denote the number of a-vertices in component $i$ at time $t$, and we
let $s_i(t)$ denote the size of component $i$ at time~$t$.  Let
$\hat{t}_i$ denote the time component $i$ is closed. We use a
shorthand notation $a_i := a_i(\hat{t}_i)$ and~$s_i := s_i(\hat{t}_i)$
for the final number of a-vertices in component $i$ and the final size
of component $i$, respectively. We also define
$n_i(t) := s_i(t) - a_i(t)$, which is the number of non-a-vertices in
component $i$ at time~$t$, and $n_i=s_i-a_i$, which is the number of
non-a-vertices in component $i$.

The high level idea behind bounding the number of advice bits used by
\VC is to prove two inequalities and then take their linear
combination. The first inequality (Lemma~\ref{lem:vc_lem1}) is more
local in that it is proved for each component independently of other
components. The second inequality (Lemma~\ref{lem:vc_lem2}) is more
global in that it incorporates potential interactions between
components. Both inequalities are proved via weight reallocation
arguments as explained in the corresponding lemmas.

We begin with the more difficult local lemma.
\begin{lemma}
\label{lem:vc_lem1}
For all $i\in[c]$, we have
\[ s_i \ge 3 a_i + 1.\]
\end{lemma}
\begin{proof}
Consider component~$i$. 

If $a_i=1$, then the vertex that received advice and its three
neighbors are added to the component by definition, so $s_i \ge 4 = 3 a_i + 1$.

If $a_i = 2$, then the two vertices that received advice can share at
most one neighbor. If, to the contrary, they had two vertices in
common, then if the first of the two vertices is rejected, then its
neighbors are accepted, and the second vertex becomes unary and does
not need advice due to P3; a contradiction.  Similarly, if the first
vertex is accepted, it becomes an aa-vertex and the second vertex gets
accepted without advice due to P6. So counting the two vertices and
their five distinct neighbors gives that $s_i \ge 7 = 3 a_i + 1$.

If $a_i \ge 3$, then the situation is more involved. The desired
inequality trivially follows from
\begin{equation}
    \label{eq:vc-target}
    n_i \ge 2a_i + 1.
\end{equation} 
For $j \in [a_i]$, let $t_j$ denote the time step when $j$th a-vertex
in component $i$ is processed. Call this vertex~$v_{j}$. Thus, the
component gets started at time $t_1$ with a-vertex~$v_1$.
Denote the three neighbors of $v_j$ by $u_{j, 1}, u_{j, 2}, u_{j, 3}$.

We prove Eq.~\eqref{eq:vc-target} using a weight reallocation
argument. Denote the weight of a vertex~$v$ by~$w(v)$. Each
non-a-vertex~$v$ that gets added to this component starts out with
weight $w(v) = 1$. Each a-vertex~$v_j$ that gets added to this
component starts out with weight $w(v_j) = 0$. The weight is
reallocated from non-a-vertices to a-vertices, so as to guarantee the
following properties at the end of processing the component:
\begin{description}
    \item[I1] the weight reallocated to the first a-vertex in the component is $1.5$: $w(v_1) = 1.5$;
    \item[I2] the weight reallocated to the second a-vertex in the component is $2.5$: $w(v_2) = 2.5$;
    \item[I3] the weight reallocated to the third a-vertex in the component is  $2.5$: $w(v_3) = 2.5$;
    \item[I4] the weight reallocated to every other a-vertex is $2$: $w(v_j) = 2$ for $j \in [4, \ell]$.
\end{description} 
Note that since we are in the case $a_i \ge 3$, this is
well-defined. We check that I1--4 are sufficient to establish the
claim. Observe that the total amount of weight allocated to component
$i$ is exactly~$n_i$. After reallocating the weight, I1--I4 imply that
the total weight in the component is
$\ge 1.5 + 2.5 + 2.5 + 2 (a_i-3) = 2 a_i + 0.5$. Since the
reallocation procedure does not destroy weight or create extra weight,
the total amount of weight in the component at the end is~$n_i$. This
implies that $n_i \ge 2a_i + 0.5$. Since $n_i$ and $a_i$ are integers,
we have $n_i \ge 2a_i + 1$, as desired.

We execute weight reallocation in parallel with \VC. The
reallocation follows some rules: (a) after sufficient weight is
reallocated to an a-vertex, this weight is not reallocated ever again;
(b) only the weights of vertices that are in component $i$ can be
reallocated (to an a-vertex in component~$i$); (c) at any point in
time, the weight of non-a-vertices can be either $0$, $0.5$, or $1$; (d)
if the weight of a non-a-vertex is $0$, then the vertex has been processed
and removed from the graph; (e) the weight of every non-a-vertex can be
reallocated twice: $0.5$ can be reallocated when its degree goes from
$3$ to $2$ (and not more than $0.5$ is reallocated in this scenario)
and the remaining $0.5$ is reallocated when the degree of the vertex
drops down further, when it is processed, or even after it is
processed; (f) every unprocessed vertex with weight $0.5$ is a
neighbor of a processed a-vertex. We do not keep track of each of the
above statements explicitly in the following case analysis, since this
is rather tedious. It is fairly straightforward to verify that each
claim continues to hold in the analysis below.

\textbf{Observation 1:} For point (b), we make one observation that is
used repeatedly, namely that a certain neighbor of a neighbor cannot
belong to an earlier component, which means that we are allowed to
reallocate weight from it. Note that the only unprocessed vertices of
a closed component are neighbors of aa-vertices. Consider an
a-vertex~$v_j$ in the current component that shares a neighbor
$u_{j,1}$ with a previously processed a-vertex~$v_{j'}$, also of
component~$i$, for some $j' < j$. Then after processing $v_j$, the
current degree of $u_{j,1}$ drops to~$1$. Let $z$ be the unique
neighbor of $u_{j,1}$ at that point. We claim that $z$ cannot belong
to a previous component. If $z$ did belong to a previous component,
say $i' < i$, then $z$ would necessarily be a neighbor of an a-vertex
in component~$i'$. Suppose that component~$i'$ was closed at
time~$t$. The degree of $u_{j,1}$ is~$3$ until $v_{j'}$ is being
processed. Thus, at time $t$, $u_{j,1}$ shared a neighbor, $z$, with
an a-vertex in component~$i'$. This implies that component $i'$ should
not have been closed at time $t$, since it could be extended by
considering~$u_{j,1}$. Thus, $z$ cannot belong to a previous
component, and we are free to allocate weight away from~$z$.

With this additional observation, we are ready to prove I1--4.

\textbf{I1.} Observe that if the first a-vertex is an ar-vertex, then
all its neighbors are removed prior to any other vertex receiving
advice. Since a-vertices in a component are connected through common
neighbors, there can be no other a-vertices in the component, so
$a_i = 1$. Therefore, since we assume that $a_i \ge 3$, the first
a-vertex must be an aa-vertex. This vertex along with its three
neighbors are added to the component. We reallocate $0.5$ unit of
weight from each of the neighbors to~$v_1$. This is illustrated in the
figure below. In order not to clutter the illustration, we do not show
all edges incident to vertices. How a vertex is processed starting at
$t_1$ is indicated next to the vertex. The weight of a vertex is shown
inside the vertex.

\begin{center}
\includegraphics[scale=0.6]{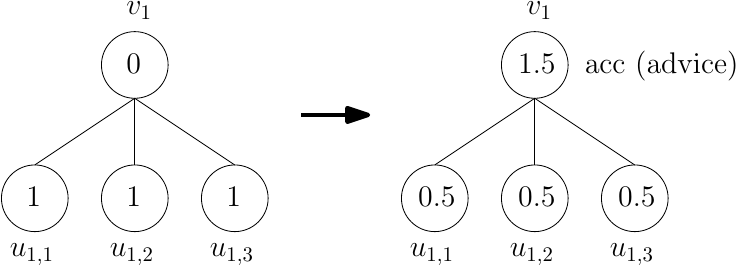}
\end{center}

\textbf{I2.} The second a-vertex~$v_2$ must have exactly one neighbor
in common with $v_1$: if it had no neighbors in common, a new
component would get started; if it had more than one neighbor in
common, then it would be processed without advice. Without loss of
generality, let that neighbor be $u_{1, 3}=u_{2,1}$. Observe that
$v_2$ must have received advice to be accepted. If it received advice
to be rejected, then all its neighbors would be accepted and $u_{1, 1}$
and $u_{1,2}$ would become a-siblings (if they have not been processed
yet), so they would get processed prior to~$v_3$. But this implies
that all neighbors of $v_1$ and $v_2$ would be eliminated prior to
$v_3$ and $v_3$ would never be added to the current component,
contradicting the assumption that $a_i \ge 3$.

Thus, we assume that $v_2$ received advice to be accepted. At time
$t_2$, the current degree of $u_{1,3}$ must be $2$: if it was higher,
then the original degree (which would include $v_1$) would be more
than $3$; if it was lower, then $u_{1,3}$ would be processed prior to
$v_2$ and $v_2$ would not have received advice. One of the vertices
contributing to the current degree of $u_{1,3}$ is~$v_2$. Let the
other vertex be~$z$. Observe that $z$ is different from all of
$u_{1,1}, u_{1,2}, u_{2,2}, u_{2,3}$ since otherwise the input graph
would contain a triangle. When $v_2$ is processed, the current degree
of $u_{1,3}$ drops to $1$, so it will be rejected and its neighbor
accepted. Since $z$ is a new vertex added to the component, we can
reallocate one unit of weight from $z$ to~$v_2$. We also reallocate
$0.5$ unit of weight from each of neighbors of $v_2$ to~$v_2$. This
results in the overall weight of $v_2$ being $2.5$, as desired. It is
easy to check that this reallocation satisfies all the rules and the
illustration is shown below.

\begin{center}
\includegraphics[scale=0.6]{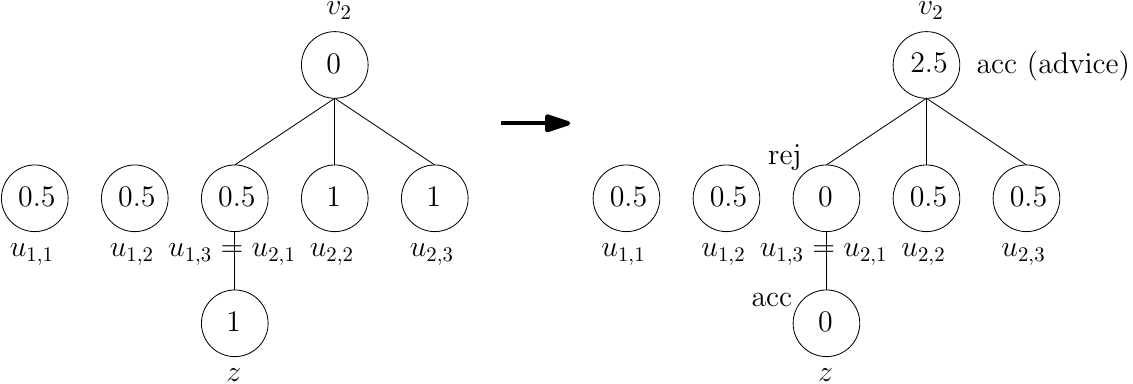}
\end{center}

\textbf{I3.} There are several cases for~$v_3$. 

\textit{Case 1.} Consider the case where $v_3$ shares a single
neighbor with a previous aa-vertex (could be either $v_1$ or
$v_2$). Without loss of generality, let the shared neighbor
be~$u_{3,1}$. Then $u_{3,2}$ and $u_{3,3}$ are added to the current
component for the first time so they start out with weight~$1$. Since
$u_{3,1}$ has not yet been processed at $t_3$, its weight is~$0.5$.

\textit{Subcase 1(a).} Suppose that $v_3$ receives advice to be
rejected. Then the weight of all its neighbors can be reallocated to
$v_3$ resulting in $w(v_3) = 2.5$, as desired. This obeys the
reallocation rules, since $v_3$ and all its neighbors will be removed
from the graph prior to~$t_4$. This is illustrated below.

\begin{center}
\includegraphics[scale=0.6]{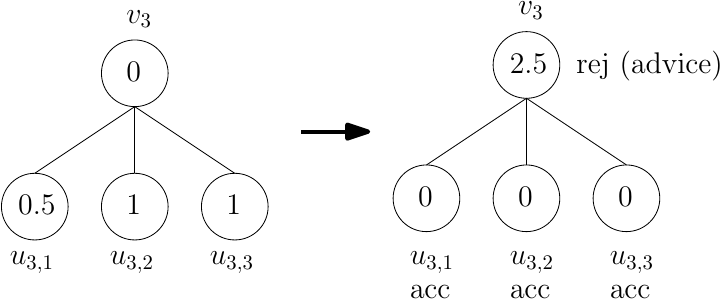}
\end{center}

\textit{Subcase 1(b).} Suppose that $v_3$ receives advice to be
accepted. Without loss of generality suppose that $u_{3,1}=u_{2,2}$,
i.e., the single shared neighbor is with $v_2$ ($v_2$ and $v_1$ behave
symmetrically in the following argument). Arguing similarly to I2,
after accepting $v_3$, the current degree of $u_{3,1}$ would drop
to~$1$. Let $z$ be the unique neighbor of $u_{3,1}$ at that
point. Then, by the priority tie breaking in P3, $u_{3,1}$ is rejected
and $z$ is accepted. If $z$ has weight $1$ at time $t_3$, then the
weight reallocation is done similarly to I2. Otherwise, $z$ has
weight~$0.5$. By Observation~1, $z$ is in the current component. The
vertex $z$ cannot be a neighbor of $v_3$, or there would be cycle. The
only vertices in the current component of weight $0.5$ after
processing vertices in I1 and I2 and reallocating weights are
neighbors of $v_1$ and~$v_2$. Since $z$ cannot be a neighbor of $v_2$
(this would create a triangle), it must be a neighbor
of~$v_1$. Without loss of generality, assume $z = u_{1,1}$. Since $z$
is accepted, $u_{1,2}$ becomes an a-sibling, unless it was already
processed. So, both $u_{1,2}$ and $u_{1,1}$ are processed and removed
from the graph prior to~$t_4$. Thus, we can reallocate $0.5$ weight
from each of $u_{1, 2}, z=u_{1,1}, u_{2, 2}=u_{3,1}, u_{3,2}, u_{3,3}$
to $v_3$ resulting in $w(v_3) = 2.5$ as desired. This last case is
illustrated below.

\begin{center}
\includegraphics[scale=0.6]{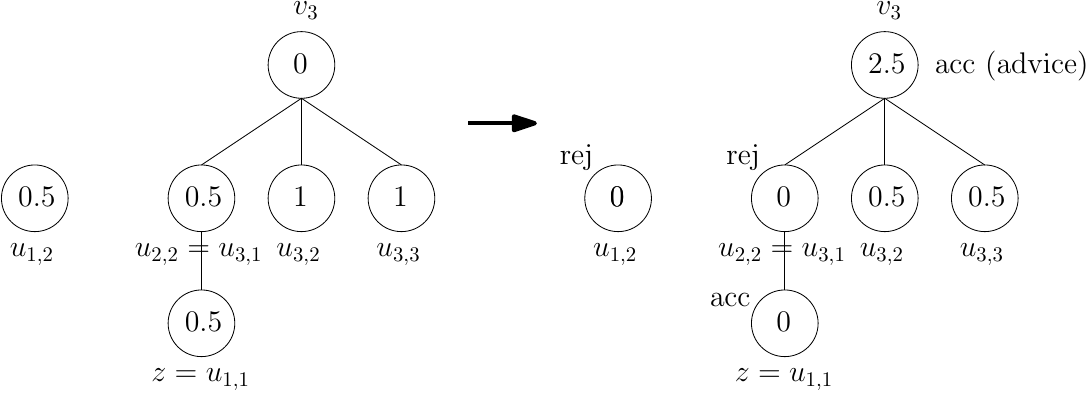}
\end{center}

\textit{Case 2.} Suppose that $v_3$ shares two neighbors with previous
aa-vertices -- one with $v_1$ and another with~$v_2$. More
specifically, without loss of generality suppose that
$u_{3,1} = u_{1,1}$ and~$u_{3,2}=u_{2,2}$.

\textit{Subcase 2(a).} If $v_3$ receives advice to be rejected, then
the three neighbors $u_{3,1}, u_{3,2}, u_{3,3}$ are accepted. Their
weights are reallocated to~$v_3$. Moreover, $u_{1,2}$ (assuming that
$u_{1,3}$ was the neighbor common to $v_1$ and~$v_2$) was either
processed earlier and had $0.5$ weight remaining, or becomes an
a-sibling and is processed prior to~$t_4$. In either case, we can
reallocate $0.5$ weight from $u_{1,2}$ to $v_3$ for the total amount
of weight reallocated to $v_3$ being~$2.5$. This is illustrated below.

\begin{center}
\includegraphics[scale=0.6]{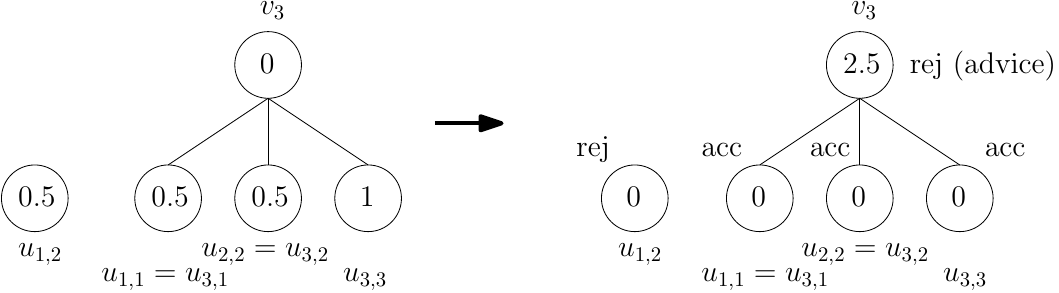}
\end{center}

\textit{Subcase 2(b).} If $v_3$ receives advice to be accepted, then
the current degrees of $u_{3,1}$ and $u_{3,2}$ drop down to $1$ each
(same argument as in I2). Let the unique neighbor of $u_{3,1}$ be
$z_1$ and the unique neighbor of $u_{3,2}$ be~$z_2$. Note that $z_1$
is not a neighbor of $v_3$ or $v_1$ in the original graph for
otherwise it would contain a triangle. Similarly, $z_2$ is not a
neighbor of $v_3$ or~$v_2$. Observe that after processing $v_3$,
vertices $u_{3,1}, u_{3,2}, z_1,$ and $z_2$ will be processed prior
to~$t_4$. If either $z_1$ or $z_2$ (which could be the same vertex)
has weight $1$ at $t_3$, then we can reallocate $0.5$ weight from each
of $u_{3,1}, u_{3,2}, u_{3,3}$ to $v_3$ and $1$ unit of weight from
$z_1$ or $z_2$ to $v_3$ for the total weight $2.5$ as desired. An
example where the weight of $z_1$ is $1$ at time $t_3$ is illustrated
below.

\begin{center}
\includegraphics[scale=0.6]{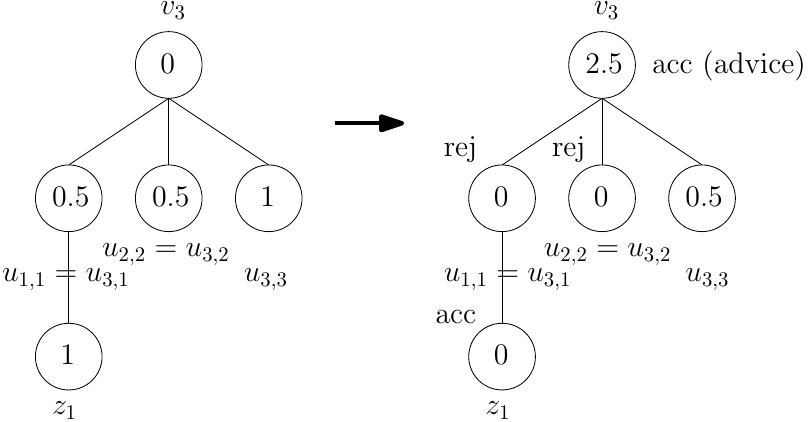}
\end{center}

The only remaining scenario is when each of $z_1$ and $z_2$ have
weight $0.5$ at time~$t_3$. Based on I1 and I2 and properties of $z_1$
and $z_2$ mentioned above, it must be the case that $z_1=u_{2,3}$
and~$z_2=u_{1,2}$, since, otherwise, there is a triangle, so
$z_1\neq z_2$. In particular, after processing $v_3$, vertices
$u_{1,1}, u_{2,2},u_{1,2},u_{2,3}$ will be processed prior
to~$t_4$. Thus, we can reallocate $0.5$ from each of them, plus $0.5$
from~$u_{3,3}$.

\textbf{I4.} Let $j \ge 4$ and consider $v_j$ receiving advice at
time~$t_j$. Each of the neighbors of $v_j$ has current degree at least
$2$ (same reason as in I2) and at least one of the neighbors is shared
with a previous aa-vertex in the component.

\textit{Case 1.} Suppose that $v_j$ receives advice to be
accepted. Without loss of generality assume that $u_{j,1}$ is a
neighbor shared with $v_{j'}$ for some $j' < j$. After processing
$v_j$ the degree of $u_{j,1}$ drops to $1$, so by the priority tie
breaking in P3, it is rejected and its neighbor, call it $z$, is
accepted. We can reallocate $0.5$ weight from each of
$u_{j,1}, u_{j,2}, u_{j,3}$ and $z$ to $v_j$ for the total weight of
$2.0$, as desired. This is illustrated below.

\begin{center}
\includegraphics[scale=0.6]{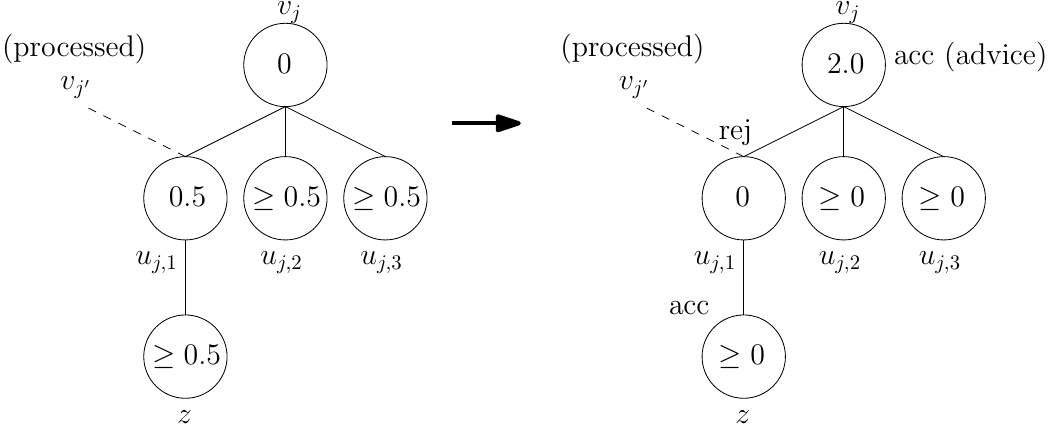}
\end{center}

\textit{Case 2.} Suppose that $v_j$ receives advice to be
rejected. Then the three neighbors are accepted. As argued before,
each of the neighbors has degree at least $2$ at time $t_j$, and each
of the neighbors has at least $0.5$ weight available for
reallocation. If at least one of the neighbors has $1$ unit of weight
available, then we can reallocate $2.0$ units of weight from the
neighbors of $v_j$ to $v_j$, as desired. If each neighbor has only
$0.5$ units available then each neighbor is also a neighbor of a
previously processed aa-vertex in this component. Let $v_{k}'$ be such
a processed neighbor of $u_{j,k}$ for $k \in [3]$. Observe that the
$v_k'$ are all distinct, since a vertex receiving advice can share at
most one neighbor with a previous aa-vertex. If some $v_k'$ is not a
contributing vertex at time $t_j$, then, by accepting $u_{j,k}$, the
other remaining neighbor of $v_k'$ becomes an a-sibling and will be
processed prior to~$t_{j+1}$. In this case, we can reallocate $0.5$
from the a-sibling and each of the $u_{j,k}$ for $k \in [3]$ to $v_j$
for a total weight of $2.0$, as desired.

The only remaining subcase is when all of the $v_k'$ are contributing
vertices at~$t_j$. This means that $v_j$ is a bad-vertex at
time~$t_j$. Consider what happens after processing~$v_j$. The degree
of each $u_{j,k}$ drops to exactly~$1$ and they are accepted. Let
$z_k$ be the unique neighbor of some $u_{j,k}$ immediately prior to
$u_{j,k}$ being accepted (note that the $z_k$ do not have to be
distinct, but it does not matter for the following argument). If,
after processing all $u_{j,k}$, the degree of at least one of the
$z_k$ drops below~$2$, then it would be processed prior
to~$t_{j+1}$. In this case, we can reallocate $0.5$ weight from each
of $u_{j,k}$ and $0.5$ weight from the to-be-processed $z_k$ to $v_j$
for a total weight of~$2.0$. Otherwise, consider $z_1$, for
example. After processing all $u_{j,k}$ the current degree of $z_1$
is~$2$. Thus, it can be rejected without advice according to
priority~P4 and its weight can be reallocated to $v_j$ for the total
weight of $v_j$ being at least $2.0$ (the other weights coming from
the $u_{j,k}$), as desired.
\end{proof}

Next, we prove the second inequality.

\begin{lemma}
\label{lem:vc_lem2}
We have
\[ 10a - 4c \le 3 n,\] where $n$ is the number of vertices in the
graph, $a$ is the number of advice bits read by \VC, and $c$
is the number of components, as defined earlier.
\end{lemma}
\begin{proof}
  We prove this via a weight reallocation argument similar to the one
  used in Lemma~\ref{lem:vc_lem1}. Weight reallocation is done in
  parallel with \VC, so we can describe it one vertex at a
  time. Weight reallocation is performed each time an input vertex
  receives advice and may involve vertices that are processed
  immediately after that without advice. There are several key
  differences from the weight reallocation done in
  Lemma~\ref{lem:vc_lem1}. First of all, every vertex starts with
  initial weight $3$ -- no matter whether the vertex is an a-vertex or
  non-a-vertex. Secondly, we allow weight to be reallocated even from
  unprocessed vertices from closed components, since we are not
  interested in a component-wise inequality, but the inequality for
  the entire input. The weight reallocation procedure will guarantee
  the following properties:
\begin{description}
\item[J1.] The first a-vertex of every component receives $6$
  units of weight.
\item[J2.] Subsequent a-vertices in every component receive $10$
  units of weight each.
\end{description}
The reallocation procedure satisfies additional constraints: (a) no
extra weight is created or consumed; (b) the weight of a vertex is at
least its current degree; (c) if a vertex has weight~$0$, then it must
have been processed; (d) at any point in time $t$ the weight that
could have been reallocated by $t$ comes only from vertices processed
by time $t$ or neighbors of a-vertices processed by time~$t$. We will
not explicitly check each of these constraints in the cases described
below, but it is easy to verify from the arguments.

We first see how J1 and J2 imply the claim and then define the
reallocation procedure to satisfy J1 and J2. Observe that after
processing the entire input, the total weight in component $i$ is at
least $6 + 10 (a_i - 1)$. Adding this over all components $i \in [c]$,
we see that the total weight in the input graph is at least
$6c + 10(a-c)$, since components are vertex disjoint. Without weight
reallocation, the total weight would be $3n$ since each vertex starts
out with exactly $3$ units of weight. Since the weight reallocation
procedure does not create extra weight, we have $3n \ge 6c + 10(a-c)$,
which implies the statement of the lemma.

Although we are allowed to reallocate weight from unprocessed vertices
from closed components, we still define the procedure for each
component separately. We use the notation of
Lemma~\ref{lem:vc_lem1}. More specifically, consider
component~$i$. Let $a_i$ denote the total number of a-vertices in the
component at the end
For $j \in [a_i]$, let $t_j$ denote the time step when the $j$th
a-vertex~$v_j$ in component~$i$ was processed. Thus, the component
gets started at time $t_1$ with a-vertex~$v_1$.
Denote the three neighbors of $v_j$ by $u_{j, 1}, u_{j, 2}, u_{j, 3}$.

\textbf{J1.} Since $v_1$ is the first vertex of the component, its
neighbors have not been processed and they cannot be neighbors of
previous a-vertices. Thus, we have
$w(v_1) = w(u_{1,1})=w(u_{1,2})=w(u_{1,3})=3$. No matter whether $v_1$
is an aa-vertex or an ar-vertex, after it is processed and removed
from the graph, the degrees of the neighbors drop by $1$ each. Thus,
we can reallocate one unit of weight from $u_{1,k}$ for $k \in [3]$
to~$v_1$, resulting in $w(v_1) = 6$, as desired. This is illustrated
below. As in Lemma~\ref{lem:vc_lem1}, we do not show all edges
incident to vertices so that the illustration does not become
cluttered. How a vertex is processed is indicated next to the
vertex. The weight of a vertex is shown inside the vertex.

\begin{center}
\includegraphics[scale=0.6]{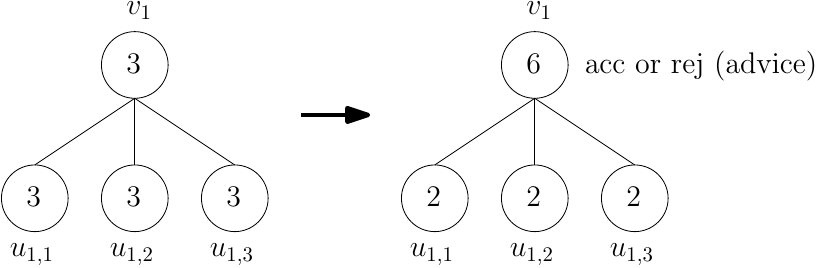}
\end{center}

\textbf{J2.} Let $j \ge 2$. We consider several cases depending on the
type of $v_j$ and its (multi-hop) neighborhood.

\textit{Case 1.} Suppose that $v_j$ receives advice to be
accepted. Since $v_j$ is not the first vertex in the component, it
shares a neighbor with a previous aa-vertex~$v_{j'}$ in the component
for some $j' <j$. Let that neighbor be~$u_{j,1}$. As in the proof of
Lemma~\ref{lem:vc_lem1}, the current degree of $u_{j,1}$ is $2$ prior
to processing $v_j$, so its weight is also~$2$. After processing
$v_j$, we reallocate $1$ unit of weight from each $u_{j,1}, u_{j,2},$
and $u_{j,3}$ to $v_j$ and the weight allocated to $v_j$
becomes~$6$. The current degree of $u_{j,1}$ drops to~$1$. Let $z$
denote the unique neighbor of $u_{j,1}$ at that moment. Then, by the
priority tie breaking in P3, $u_{j,1}$ is rejected and $z$ is
accepted. We reallocate one additional unit of weight from $u_{j,1}$
to~$v_j$. Since $z$ was present in the graph prior to $v_j$ being
processed, the current degree of $z$ at time $t_j$ must be at
least~$2$. After $z$ is processed, we reallocate its weight
to~$v_j$. At this point, the weight allocated to $v_j$ becomes at
least~$9$. Let $y$ be any neighbor of $z$ other than $u_{j,1}$ prior
to $z$ being removed. Since processing $z$ decreases the degree of $y$
and we do not care which component $y$ belongs to, we reallocate one
unit of weight from $y$ to $v_j$ resulting in total weight allocated
to $v_j$ being~$10$. Observe that the triangle-free condition ensures
that $z$ is not $u_{j,2}, u_{j,3}$ and it does not matter for the
argument whether $y$ is $u_{j,2}$ or $u_{j,3}$ or any other vertex in
the graph. This case is illustrated below.

\begin{center}
\includegraphics[scale=0.6]{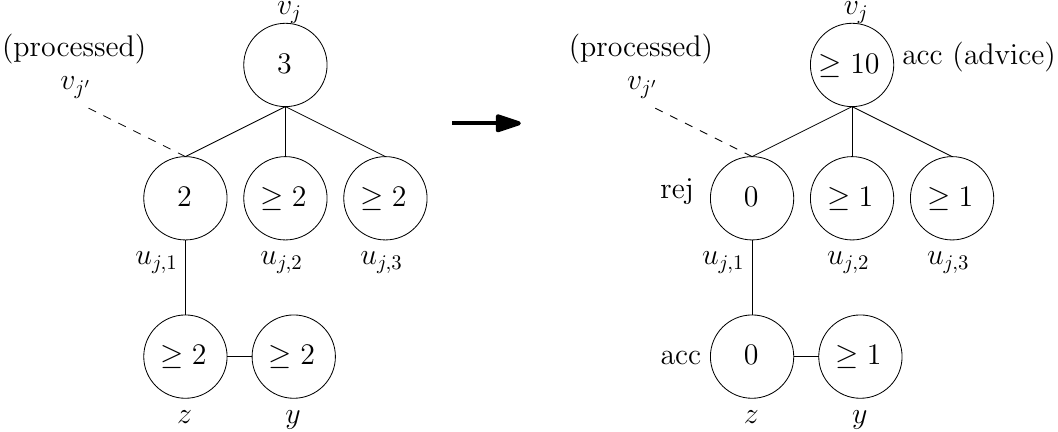}
\end{center}

\textit{Case 2.} Suppose that $v_j$ receives advice to be
rejected. All neighbors of $v_j$ will be accepted after that and we
can reallocate the weight from those neighbors to~$v_j$. The current
degree of each neighbor of $v_j$ is at least $2$ prior to $v_j$ being
processed (see arguments in Lemma~\ref{lem:vc_lem1} for why). Thus, if one of the neighbors has current weight $3$, then the
total weight reallocated to $v_j$ from its neighbors is at
least~$7$. This, together with $v_j$'s initial weight of~$3$,
results in $w(v_j) \ge 10$, as desired.

It only remains to handle the case when neighbors of $v_j$ have
current degree and weight $2$ at the time $v_j$ is processed. Let $z$
be the unique neighbor of~$u_{j,1}$. The current degree of $z$ is at
least $2$ prior to $v_j$ being processed, so its weight is at least
$2$, as well. Processing $v_j$ and its neighbors decreases the degree
of $z$ by at least~$1$ and therefore we may to reallocate one unit of
weight from $z$ to~$v_j$. This last case is illustrated below.

\begin{center}
\includegraphics[scale=0.6]{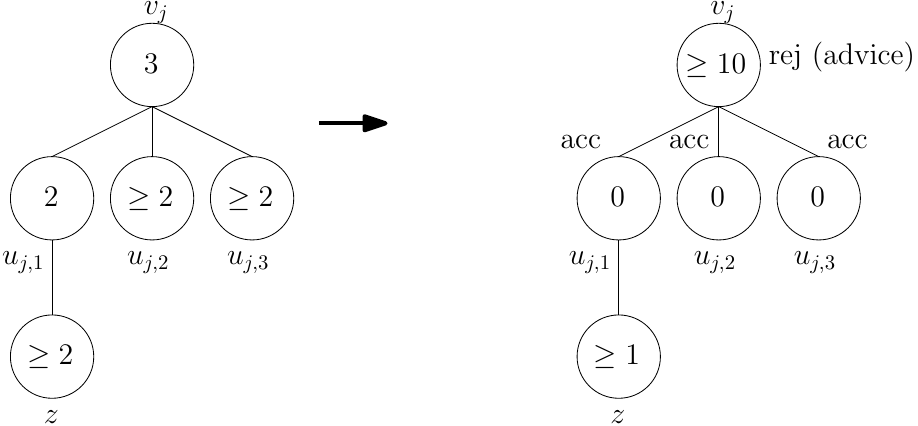}
\end{center}
\end{proof}

We are now ready to prove the bound on the number of advice bits used
by \VC.

\begin{lemma}
\label{lem:vc-adv-len}
\VC uses at most $(7/22)n = 0.31\overline{81} n$ bits of
advice on triangle-free graphs of maximum degree~$3$.
\end{lemma}
\begin{proof}
  Lemma~\ref{lem:vc_lem1} says that $1 + 3 a_i \le s_i$ for
  $i \in [c]$. Since the components are vertex-disjoint, the total
  number of vertices that received advice is $a = \sum_{i=1}^c a_i$
  and the total number of vertices is $n = \sum_{i=1}^c s_i$. Adding
  these inequalities over all $i \in [c]$, we obtain
\begin{equation}
    \label{eq:vc-first}
    3a+c \le n
\end{equation}
Lemma~\ref{lem:vc_lem2} says that
\begin{equation}
    \label{eq:vc-second}
    10a-4c \le 3n
\end{equation}
Adding $4$ times Eq.~\eqref{eq:vc-first} to Eq.~\eqref{eq:vc-second} results in $22 a \le 7 n$, i.e., $a \le (7/22) n$, as desired.
\end{proof}

\begin{corollary}
  The priority exact algorithm corresponding to \VC runs in time
  \[ \lOh{2^{\frac{7 n}{22}}}\subset\Oh{1.247^n}.\]
\end{corollary}

\section{Hardness Results Using Templates}
\label{sec:template}

In this section, we present templates for proving lower bounds on how
much advice is needed for an adaptive priority algorithm to achieve a
certain approximation ratio or optimality. The results hold in the
oblivious priority function model (and the optimality results also
hold in the decision-based priority function model).

The rest of this section is organized as follows: In
Subsection~\ref{ssec:gadget_pairs}, we introduce the notion of gadget
pattern pairs and describe conditions on problems and gadget pattern
pairs that are sufficient for proving lower bounds using the templates
in the next two subsections.
In Subsection~\ref{ssec:lb_approx}, we present
templates for proving trade-offs between the number of advice bits and
\emph{approximation} ratios. We finish the section with a table
listing the lower bound results that can obtained for Minimum Vertex Cover
with the gadget pattern pairs from Subsection~\ref{sec:firstexample} and with
known gadget pattern pairs for five other problems.
In Subsection~\ref{sec:lbopt}, we present the template for proving lower
bounds on the number of advice bits needed to solve problems to
\emph{optimality}. The implications of these results for priority exact
algorithms are also discussed.

\subsection{Gadget Pattern Pairs for the Templates}
\label{ssec:gadget_pairs}

In this section, we generalize the construction introduced in
Section~\ref{sec:firstexample}. These types of constructions will be
used in our lower bound proofs, some based on reductions and some
adversarial.  Thus, in some proofs, vertices are given to the priority
algorithm with advice by
an adversary and, in other proofs, by a reduction (algorithm).  In this
section, we just use the term ``adversary'' to represent both of
these options.

In Section~\ref{sec:firstexample}, we presented a lower bound on
solving the Minimum Vertex Cover problem to optimality using priority
algorithms with advice in the decision-based priority function model.
Two graphs, Graph~1 and Graph~2 were used. When a vertex of degree~$2$
was selected, the adversary chose between two isomorphic copies of
Graph~1 to include; these two isomorphic copies constitute an example
of the general concept, a \emph{gadget pattern pair}. Similarly, for
a vertex of degree~$3$, the
isomorphic copy of Graph~1, along with the isomorphic copy of Graph~2,
was another example of a gadget pattern pair. These two gadget pattern pairs
constitute our \emph{collection of gadget pattern pairs}
for the Minimum Vertex Cover problem.

A gadget $G$ for problem $B$ is simply some constant-sized instance for
$B$, i.e., a collection of input items that satisfy the consistency
conditions for problem~$B$.  For example, if $B$ is a graph problem in
the vertex arrival, vertex adjacency model, $G$ could be a
constant-sized graph.  In this case, an input item would possibly be a vertex
name and a list of neighboring vertex names.

We will define a universe of input items from a union of subuniverses.
For this graph problem, in a subuniverse for a collection of gadget
pattern pairs, each vertex name exists many times as the vertex of an
input item in the universe, because it can be paired with many
different possible lists of neighboring vertex names for the purpose
of making all possible isomorphic instances of the gadget.  The effect
of this is that when an algorithm receives the first input item of
some degree~$d$, it can be any of the degree~$d$ vertices in any
of the gadget patterns in the collection.
% , i.e., the algorithm does not learn anything from the names of vertices. 
Consistency conditions must apply to the actual given input. For
instance, for each vertex name $u$ which is listed as a neighbor of
$v$, it must be the case that $v$ is listed as a neighbor
of~$u$. There could of course be further constraints on the input
instances; for instance, restricting inputs to graphs of some maximum
degree.

In our proofs, the adversary provides multiple gadgets (possibly many
isomorphic ones), each coming from some gadget pattern pair in the
collection.  We need that the sets of possible input items for these
multiple gadgets are disjoint, but contain all necessary input
items for all gadget patterns in the collection of gadget pattern pairs.
To obtain this, we repeat the
construction above, creating distinct subuniverses for each gadget the
adversary presents.  Thus, if, during the execution of an algorithm,
the adversary presents $m$ gadgets to the algorithm, the universe
 consists of $m$ disjoint subuniverses,
$\mathcal{U}_1, \mathcal{U}_2, \ldots, \mathcal{U}_m$; all of these
subuniverses are identical up to renaming of vertices.
This implies that an input item identifies which subuniverse it is in.
We refer to this property as the \emph{disjoint copies condition}.

We also make an assumption on the objective function related to the
gadgets: We say that the objective function for a problem $B$ is
\emph{additive with respect to the gadgets}, if, for any instance
formed from a set of $m$ gadgets from disjoint universes, the
objective function value on the instance is the sum of the objective
function values on the individual gadgets. This implies that
optimality on the instance requires optimality on each gadget. For
example, this assumption will hold for many classical graph problems
since the gadgets will be maximal connected components and the
corresponding objectives are additive with respect to connected
components.

Recall that $\max_P R$ denotes the first item in a set $R$ according
to the current priority function $P$, i.e., the highest priority item
(possibly after tie-breaking by an adversary).
Assume that $\ALG$ responds ``accept'' or ``reject'' to any possible
input item.  This captures problems such as Minimum Vertex Cover, Independent
Set, Clique, etc.

Each collection of gadget pattern pairs also satisfies the first item
condition, and the distinguishing decision condition. The \emph{first
  item condition} says that the first input item chosen by $\ALG$ from
the subuniverse $\mathcal{U}_j$, $\first{\mathcal{U}_j}$, identifies a
gadget pattern pair, $(G_j^a,G_j^r)$, from the collection of gadget
pattern pairs, and that the input item itself gives no information
about which of the two gadgets $G_j^a$ or $G_j^r$ it is in.
For the Vertex Cover example from Section~\ref{sec:firstexample},
the first item could
be a vertex of degree~$2$ or degree~$3$, and the two cases lead to
different gadget pattern pairs, but the actual input item gives no information
as to which of the gadget patterns within the pair it belongs to.
Given a priority function $P$, the first item condition can be written
as: $\first{\mathcal{U}_j}=\max_P G_j^a = \max_P G_j^r$.
The \emph{distinguishing decision condition} says that the decision
with regards to item $\first{\mathcal{U}_j}$ that results in the
optimal value of the objective function in $G_j^a$ is different from
the decision that results in the optimal value of the objective
function in $G_j^r$.  This first input item is said to be the
\emph{distinguishing item}.
For
accept/reject, we list $G_j^a$, where the correct decision is to
accept, as the first gadget pattern of the pair, and $G_j^r$ as the second.

\subsection{Lower Bounds on the Advice Needed for Approximation}
\label{ssec:lb_approx}
In this section, we establish two theorems that give general templates
for gadget-based reductions from a problem referred to as 2-SGKH, one
for maximization problems and one for minimization problems.  While it
takes some work to establish these results, the theorems are easy to
apply to concrete problems once established.  One simply has to
define a collection of gadget pattern pairs with the required properties
and then plug numbers
into our formulas.  We do this for a number of approximation problems
at the end of this section.

The following online problem, while seeming artificial, has been used
extensively in proving lower bounds for online algorithms with advice,
and we can also use it for adaptive priority algorithms with advice.

\begin{definition}
  The \emph{Binary String Guessing Problem~\cite{BHKKSS14} with known
    history} (2-SGKH) is the following online problem. The input
  consists of $(n, \sigma=(x_1, \ldots, x_n))$, where
  $x_i \in \{0,1\}$. Upon seeing $x_1, \ldots, x_{i-1}$, an algorithm
  guesses the value of $x_i$. The actual value of $x_i$ is revealed
  after the guess. The goal is to maximize the number of correct
  guesses.
\end{definition}

B\"{o}ckenhauer et al.~\cite{BHKKSS14} provide a trade-off between the
number of advice bits and the approximation ratio for the binary
string guessing problem. This can be used to show that a linear number
of bits of advice are necessary for many online problems.

\begin{theorem}[B\"{o}ckenhauer et al.~\cite{BHKKSS14}]
\label{thm:2sgkh-lb}
For the $2$-SGKH problem and any $\epsilon\in\EPSINT$, no online
algorithm using fewer than $(1-H(\epsilon))n$ advice bits can make
fewer than $\epsilon n$ mistakes for large enough~$n$, where
$H(p) = H(1-p) = - p \log (p) - (1-p) \log (1-p)$ is the binary
entropy function.
\end{theorem}

To obtain an optimal online algorithm with advice for $2$-SGKH,
$n$~bits of advice are necessary and sufficient~\cite{BHKKSS14}.

Results and proofs presented here are somewhat similar to those
presented in~\cite{BBLP20} for fixed priority algorithms with
advice. However, there are two major differences.  The harder and more
interesting one is that we handle adaptive priorities, where the
priority functions
may depend partially on the advice. 
In addition, we reduce from string guessing directly instead of
going via an intermediate priority algorithm problem. The purpose of
this is to avoid losing constant factors with regards to the
inapproximability results through intermediate reductions, but this
change also made it easier to handle adaptive priorities.

The lower bounds in this section hold in the oblivious priority function
model. Recall that in Section~\ref{sec:firstexample}, we showed a lower
bound result for solving Vertex Cover to optimality in the decision-based
priority function model. It is an open problem to determine if the
approximation lower bounds we prove here also hold in the decision-based
priority function model. The problem in proving this when dealing
with approximation algorithms is that, theoretically,
a priority algorithm with advice could use that advice to encode
information in the decisions it makes and then use those decisions
in later priority functions. This would allow the priority functions
to depend on the advice. For algorithms solving a problem to optimality,
this encoding cannot be done since 
the gadgets in the proof ensure that the each decision made by an
optimal algorithm is forced.

The template is restricted to binary decision problems since the goal
is to derive inapproximability results based on the 2-SGKH problem,
where guesses (answers) are either $0$ or $1$.  In our reduction from
2-SGKH to a problem $B$, we assume that we have a priority algorithm
$\ALG$ with advice in the oblivious priority function model for problem $B$.
Thus, the
priority functions may vary between inputs to $\ALG$, but are oblivious
when the input item selected has no apparent relation to any input seen
before that point. The
current priority function will generally be referred to as $P$.
For the reduction, the inputs to 2-SGKH are
$X=\langle x_1, \ldots, x_n\rangle$.

\paragraph{Reduction algorithm}  Based
on $\ALG$, its advice, and its priority functions, we define an online
algorithm $\ALG'$ with advice (the reduction algorithm) for 2-SGKH.
The reduction is \emph{advice-preserving}, since $\ALG'$ only uses the
advice that $\ALG$ does, no more.  The input items,
$n,x_1,x_2,\ldots,x_n$ with $x_i \in \{ 0,1\}$, to 2-SGKH arrive in an
online manner, so after $n$ arrives, $\ALG'$ must guess $x_1$, and
then the actual value of $x_1$ is revealed. 
In the general case,
immediately after the value $x_i$ is revealed, $\ALG'$ must guess
$x_{i+1}$ and then the actual value $x_{i+1}$ is revealed.  When $x_n$
is revealed, $\ALG'$ knows that this is the end of the input. At the
end, there is some post-processing to allow $\ALG'$ to complete its
computation.
$\ALG'$ is outlined in Algorithm~\ref{alg:template}, but we
now describe how $\ALG'$ provides input to \ALG in a consistent manner.
\begin{algorithm}
\caption{The reduction algorithm.}\label{alg:template}
\textbf{Given:} $\ALG$ for problem~$B$; the inputs to 2-SGKH are
$X=\langle x_1, \ldots, x_n\rangle$ \\

\begin{algorithmic}[1]
  \State{$R=\mathcal{U}$}
%  \State{Let $R$ be the set of remaining items defined by and updated
%    by the adversary.}
  \Comment{Use the input to $B$ to give answers for $X$}
\State $i=0$  \Comment{Current index of 2-SGKH input}
\While {$i < n$}  
  \State{Let $P$ be the current priority function for $\ALG$} %\Comment{Update the priority function}
  \State $v = \max_P R$ \Comment{Choose $v$ as described in~\ref{consistent-choice}}
  \If {$v$ is the first vertex from universe $\mathcal{U}_{i+1}$}
  \State{$i=i+1$}
        \State{present $v$ to $\ALG$} \label{line-present-mone}
        \State{answer $0$ if $\ALG$ answers ``accept'' and $1$ if $\ALG$ answers ``reject''} \label{line-answer-same}
        \State{receive actual $x_i$}
        \State{update $R$ to only contain vertices from one of the two gadgets}
        \State{make $H_i$ gadget $G_i^a$ if $x_i=0$ and $G_i^R$ if $x_i=1$}
        \State{$R=R\setminus (\mathcal{U}_i\setminus H_i)$}
        \Else
        \State{present $v$ to $\ALG$}
        \EndIf  
        \State $R=R \setminus \SET{v}$
\EndWhile
\While {$R\not=\emptyset$}      \Comment{Post-processing to finish inputs for problem $B$}
  \State Let $P$ be the current priority function for $\ALG$
  \State $v= \max_P R$ 
  \State present $v$ to $\ALG$
  \State $R=R \setminus \SET{v}$
\EndWhile
\end{algorithmic}
\end{algorithm}

\paragraph{Consistent choice of input items}
\label{consistent-choice}
\begin{itemize}
  \item
    $\ALG'$ defines the universe $\mathcal{U}$ to be the union of $n$ disjoint gadget
    pair universes,
    $\{ \mathcal{U}_1, \mathcal{U}_2,\ldots,\mathcal{U}_n\}$.
  It eventually defines an input to problem $B$, $H_1,H_2,\ldots,H_k$,
  where $H_i$ is a gadget $G_i^a$ from $\mathcal{U}_i$ if $x_i=0$;
  otherwise it is a gadget $G_i^r$ from $\mathcal{U}_i$.

  These $H_i$ can be defined initially, if the input items are
  isomorphic, in which case a set $R$ is initialized to contain the
  input items from these gadgets.  Otherwise, as in the case of the
  Vertex Cover gadget patterns from Section~\ref{sec:firstexample}, the
  algorithm's priority functions can give subsets of input items with
  identical priorities due to its oblivious nature.
  % I find this parenthesis more confusing than helpful.
  %(in the case of
  %Vertex Cover, the vertices of degree~2 and the vertices of degree~3
  %could be in different sets).
  Knowing the inputs to 2-SGKH and using the fact that the first item
  condition holds,
  $\ALG'$ can always determine
  which gadget to actually use for $H_i$ when the first input item from
  $\mathcal{U}_i$ is selected. The set $R$ initially
  contains all of the universe $\mathcal{U}$, and $\ALG'$ removes input
  items from $R$ that are in $\mathcal{U}_i$, but are not in $H_i$,
  when $H_i$ has been determined.
  Other input items are removed as they are processed.

\item
  $\ALG'$ decides which input item to give when the
  algorithm's priority function designates a set $S$ of size greater
  than one as those input items having highest priority. If at least
  one input item from every universe has been processed, the reduction
  algorithm
  can make an arbitrary choice, lexicographically, for example. The
  same holds if the input items contain names of one or more input
  items that have already appeared in earlier input items (for a graph
  in the vertex arrival, vertex adjacency model, this means that the
  input item is a neighbor or a neighbor of a neighbor of some vertex
  already processed.) Otherwise, $\ALG'$ has arranged that \ALG
  has seen or will see input items from the first $i-1$ universes and
  now presents the first from $\mathcal{U}_i$. From the set of input
  items with current highest priority, $\ALG'$ chooses which
  gadget pattern is correct for $H_i$: $G_i^a$ if $x_i=0$ or $G_i^r$
  if $x_i=1$, satisfying that the distinguishing item, $v$, for the
  gadget is among those in~$S$. $\ALG'$ presents $v$ to the
  algorithm and chooses the actual gadget $H_i$ consistent with that.
\end{itemize}

The main challenge is to ensure that the
input items to $\ALG$ are presented in the order determined by the
priority functions, which may change over time.  The fact that the
priority function does not distinguish between input items that
have no known connection to input items already seen allows $\ALG'$
to choose a distinguishing item in a new gadget from a new universe
when that is necessary. In this case, by the disjointness of the
universes for the
gadgets and the obliviousness of the priority functions, such a
distinguishing item will always be in the set of items of highest
priority.
Thus, the first items in the successive gadgets are chosen
in order. The first item chosen from a gadget is one where the
distinguishing decision condition holds, i.e., one where one decision is
optimal for that gadget and the other leads to a non-optimal solution.

We let $\ALG(I)$ denote the value of the objective function for $\ALG$
on input $I$. The size of a gadget pattern~$G$, denoted by $\SIZE{G}$, is the
number of input items specifying a gadget consistent with that gadget
pattern. We write $\OPT(G)$ to
denote the best value of the objective function on $G$. Recall that we
focus on problems where a solution is specified by making an
accept/reject decision for each input item. We slightly abuse notation
and let $\first{G}$ denote
the input item from gadget~$G$ that was presented to $\ALG$ first, due
to $\ALG'$'s choice among the set of input items with highest priority.
We write $\BAD(G)$ to denote the best
value of the objective function attainable on $G$ after making the
wrong decision for that first item, $\first{G}$, i.e., if there is an
optimal solution that accepts (rejects) $\first{G}$, then $\BAD(G)$
denotes the best value of the objective function \emph{given} that
$\first{G}$ was rejected (accepted).

\begin{theorem}
\label{thm:template-minimization}
Consider a collection of $k\geq 1$ gadget pattern pairs
$\SETOF{(G_j^a, G_j^r)}{1\leq j\leq k}$ for a minimization problem~$B$.
Suppose that the objective function for $B$
is additive with respect to the gadgets.
Let $\mathcal{U}_1, \ldots, \mathcal{U}_n$ each be subuniverses
for the collection of gadget pattern pairs,
and the universe $\mathcal{U}$ be the union of the subuniverses.
Assume that the following
conditions are satisfied for the gadget pattern pairs
with respect to the subuniverses: the consistency
condition for $B$, the first item condition, the distinguishing
decision condition, and the disjoint copies condition.
Let $s$ be
the maximum number of input items in any gadget pattern in the collection.
Suppose the values
$$\OPT(G_j^a), \BAD(G_j^a), \OPT(G_j^r), \BAD(G_j^r)$$ are independent
of $j$, and we denote them by
$$\OPT(G^a), \BAD(G^a), \OPT(G^r), \BAD(G^r);$$
we assume that $\OPT(G^r) \ge \OPT(G^a)$.  Define
$\rho=\min\SET{\frac{\BAD(G^a)}{\OPT(G^a)},\frac{\BAD(G^r)}{\OPT(G^r)}}$.
Then for any $\epsilon \in\EPSINT$, no adaptive priority
algorithm in the oblivious priority function model
using fewer than $(1-H(\epsilon))n/s$ advice bits can
achieve an approximation ratio smaller than
\[1 + \frac{\epsilon (\rho-1) \OPT(G^a)}{\epsilon \OPT(G^a) + (1-\epsilon)\OPT(G^r)}.\]
\end{theorem}
\begin{proof}
  Consider an adaptive priority algorithm \ALG for~$B$
  in the oblivious priority function model.
  A reduction from 2-SGKH is specified in Algorithm~\ref{alg:template},
    combined with the definition of $\ALG'$.
  The set $R$ contains the remaining items which could still be in
  the input to $B$ and have not yet been presented to \ALG.    At
  any point in time, one of the input items with the highest priority among
those
still available in
$R$ is presented to $\ALG$. This item is the first input item
from a gadget when (1) there are still gadgets in $R$, where none of their
input items have been seen, and (2) the set of input items with highest
priority is the set of input items containing no reference to any input
item referenced in any input item already seen.
If this item is the first
input item from a gadget,  $H_i$, from $\first{\mathcal{U}_i}$, it is an input
item where the distinguishing decision condition holds.
In this case, the next input to 2-SGKH to be
processed is $x_i$, and $\ALG'$ guesses $0$ for $x_i$ if $\ALG$ accepts
$\first{\mathcal{U}_i}$ and $1$ if $\ALG$ rejects.
Note that $\ALG'$ has created $H=\langle H_1,H_2,\ldots,H_n\rangle$
such that the answer $\ALG$ gives is correct for problem $B$ if
and only if the answer $\ALG'$ gives is correct for 2-SGKH.
The correct answer by \ALG is
well defined by the distinguishing decision condition.

The amount of advice is the same for both algorithms, so when it is
$(1-H(\epsilon))n'$ bits for the $n'$ inputs to 2-SGKH, it is at least
$(1-H(\epsilon))n/s$ bits for the $n\leq sn'$ inputs to $B$.

Now we turn to the approximation ratio obtained. We want to
lower-bound the number of incorrect decisions by $\ALG$. We focus on
the input items which are $\first{\mathcal{U}_i}$ and assume that
$x_i$ is the next input to 2-SGKH when $\first{\mathcal{U}_i}$ is
processed.  Assume that $\ALG$ answers correctly on all inputs that
are not $\first{\mathcal{U}_i}$ for any $i$.

We know from Theorem~\ref{thm:2sgkh-lb} that for any
$\epsilon \in \EPSINT$, any online algorithm using fewer than
$(1-H(\epsilon))n$ advice bits makes at least $\epsilon n$ mistakes on
2-SGKH. Since we want to lower-bound the approximation ratio of
$\ALG$, and since a ratio larger than one decreases when increasing
the numerator and denominator by equal quantities, we can assume that
when $\ALG$ answers correctly, it is on the gadget pattern pair with the larger
$\OPT$-value, $G^r$. For the same reason, we can assume that the ``at
least $\epsilon n$'' incorrect answers are in fact exactly
$\epsilon n$, since classifying some of the incorrect answers as
correct just lowers the ratio. For the incorrect answers, assume that
the gadget pattern $G^a$ is presented $w$ times, and, thus, the gadget pattern, $G^r$,
$\epsilon n - w$ times.

Denoting the input created by $\ALG'$ for $\ALG$ by $I$, we obtain the
following, where we use that $\BAD(G_j^x) \geq \rho\OPT(G_j^x)$ for
$x\in \{ a,r\}$.  Since the objective function for problem $B$ is
additive,

\[\begin{array}{rcl}

\dfrac{\ALG(I)}{\OPT(I)} & \geq &

\dfrac{(1-\epsilon)n\OPT(G^r)+w\BAD(G^a)+(\epsilon n - w)\BAD(G^r)}{
  (1-\epsilon)n\OPT(G^r)+w\OPT(G^a)+(\epsilon n - w)\OPT(G^r)}
\\[4ex]
& \geq &

\dfrac{(1-\epsilon)n\OPT(G^r)+w\rho\OPT(G^a)+(\epsilon n - w)\rho\OPT(G^r)}{
  (1-\epsilon)n\OPT(G^r)+w\OPT(G^a)+(\epsilon n - w)\OPT(G^r)}
\\[4ex]
& = &

1 + \dfrac{w(\rho-1)\OPT(G^a) + (\epsilon n -w)(\rho-1)\OPT(G^r)}{
  w\OPT(G^a) + (n-w)\OPT(G^r)}

\end{array}\]

Taking the derivative with respect to $w$ and setting equal to zero
gives no solutions for $w$, so the extreme values must be found at the
endpoints of the range for $w$ which is $[0,\epsilon n]$.

Inserting $w=0$, we get $1+\epsilon(\rho-1)$, while $w=\epsilon n$ gives
\[1 + \frac{\epsilon (\rho-1) \OPT(G^a)}{\epsilon \OPT(G^a) + (1-\epsilon)\OPT(G^r)}.\]

The latter is the smaller ratio and thus the lower bound we can provide.
\end{proof}

The following theorem for maximization problems is proved analogously.

\begin{theorem}
\label{thm:template-maximization}
Consider a collection of $k\geq 1$
gadget pattern pairs $\SETOF{(G_j^a, G_j^r)}{1\leq j\leq k}$ for a maximization problem~$B$,
satisfying the conditions in Theorem~\ref{thm:template-minimization}.
Let $s$ be the maximum number of input items in any gadget pattern in
the collection.  Then for any $\epsilon \in \EPSINT$, no adaptive
priority algorithm using fewer than $(1-H(\epsilon))n/s$ advice bits
can achieve an approximation ratio smaller than
\[1 + \frac{\epsilon (\rho-1) \OPT(G^a)}{\epsilon \OPT(G^a) + (1-\epsilon)\rho\OPT(G^r)}\]
where
$\rho=\min\SET{\frac{\OPT(G^a)}{\BAD(G^a)}, \frac{\OPT(G^r)}{\BAD(G^r)}}$.
\end{theorem}

\begin{proof}
  The proof proceeds as for the minimization case in
  Theorem~\ref{thm:template-minimization} until the calculation of the
  lower bound of $\frac{\ALG(I)}{\OPT(I)}$.  We continue from that
  point, using the inverse ratio to get values larger than one.

We use that for $x\in \{ a,r\}$,  $\BAD(G^x) \leq \OPT(G^x) / \rho$.
\[\begin{array}{rcl}

\dfrac{\OPT(I)}{\ALG(I)} & \geq &

\dfrac{(1-\epsilon)n\OPT(G^r)+w\OPT(G^a)+(\epsilon n - w)\OPT(G^r)}{
      (1-\epsilon)n\OPT(G^r)+w\BAD(G^a)+(\epsilon n - w)\BAD(G^r)}
\\[4ex]
& \geq &

\dfrac{(1-\epsilon)n\OPT(G^r)+w\OPT(G^a)+(\epsilon n - w)\OPT(G^r)}{
      (1-\epsilon)n\OPT(G^r)+\frac{w}{\rho}\OPT(G^a)+\frac{\epsilon n - w}{\rho}\OPT(G^r)}

\end{array}\]

Again, taking the derivative with respect to $w$ gives an always
non-positive result. Thus, the smallest value in the range
$[0,\epsilon n]$ for $w$ is found at $w=\epsilon n$.  Inserting this
value, we continue the calculations from above:

\[\begin{array}{rcl}

\dfrac{\OPT(I)}{\ALG(I)} & \geq &
\dfrac{(1-\epsilon)n\OPT(G^r)+w\OPT(G^a)+(\epsilon n - w)\OPT(G^r)}{
      (1-\epsilon)n\OPT(G^r)+\frac{w}{\rho}\OPT(G^a)+\frac{\epsilon n - w}{\rho}\OPT(G^r)}
\\[4ex]

& = &
\dfrac{(1-\epsilon)n\OPT(G^r)+(\epsilon n)\OPT(G^a)}{
      (1-\epsilon)n\OPT(G^r)+\frac{\epsilon n}{\rho}\OPT(G^a)}
\\[4ex]

& = &
\dfrac{(1-\epsilon)\rho\OPT(G^r)+\epsilon \rho\OPT(G^a)}{
      (1-\epsilon)\rho\OPT(G^r)+\epsilon\OPT(G^a)}
\\[4ex]

& = &
1 + \dfrac{\epsilon (\rho-1) \OPT(G^a)}{
          (1-\epsilon)\rho\OPT(G^r)+\epsilon\OPT(G^a)}

\end{array}\]
The latter is the smaller ratio and thus the lower bound we can
provide.
\end{proof}

We mostly use Theorems~\ref{thm:template-minimization}
and~\ref{thm:template-maximization} in the following specialized form.

\begin{corollary}
\label{corollary-form-used}
With the setup from Theorems~\ref{thm:template-minimization}
and~\ref{thm:template-maximization}, we have the following:

For a minimization problem, if
$\OPT(G^a)=\OPT(G^r)=\BAD(G^a)-1=\BAD(G^r)-1$, then no adaptive
priority algorithm using fewer than $(1-H(\epsilon))n/s$ advice bits
can achieve an approximation ratio smaller than
$1+\frac{\epsilon}{\OPT(G^a)}$.

For a maximization problem, if
$\OPT(G^a)=\OPT(G^r)=\BAD(G^a)+1=\BAD(G^r)+1$, then no adaptive
priority algorithm using fewer than $(1-H(\epsilon))n/s$ advice bits
can achieve an approximation ratio smaller than
$1+\frac{\epsilon}{\OPT(G^a)-\epsilon}$.
\end{corollary}

\bigskip
For the Minimum Vertex Cover problem, for example, we can apply the
minimization version of Corollary~\ref{corollary-form-used}. The size
of the gadget patterns is $s=7$ vertices in all cases. Since
$\OPT(G^a)=\OPT(G^r) =3$, and, when the optimal decision is not made
on the first vertex processed, the vertex cover size is at least $4$,
we obtain the following:

\begin{corollary}
\label{vclb}
For Minimum Vertex Cover and any $\epsilon\in\EPSINT$, no adaptive
priority algorithm using fewer than $(1-H(\epsilon))n/7$ advice bits
can achieve an approximation ratio smaller than
$1+\frac{\epsilon}{3}$.
\end{corollary}

\bigskip
The gadget pattern pairs used in~\cite{BBLP20} (called gadget patterns
in that paper) to prove lower bounds in the
fixed priority model also work here in the adaptive priority model;
there are no additional restrictions used in the proof here. (These
gadget patterns are included in the appendix for completeness.) The
reductions done here are directly from 2-SGKH, as opposed to going
through the Pair Matching problem, as in~\cite{BBLP20}. As mentioned
earlier, this
makes the proofs simpler in most respects (except for having to take
into account changing priority functions), and it means that one does
not lose a factor $2$ in the amount of advice required.  Thus, the
results from~\cite{BBLP20} can be expressed using
Table~\ref{table:summary} as adaptive priority algorithm with advice
lower bounds.  All of the ratios obtained approach one as the amount
of advice approaches some fraction of $n$. The gadget pattern pairs used
in~\cite{BBLP20} can also be used for lower bounds on the amount of
advice required for optimality. Thus, those gadget pattern pairs satisfy the
conditions of both templates in this paper.

To collect results in one table, we include results for optimality
though they are treated in the next section. 

%\newcolumntype{P}[1]{>{\centering\arraybackslash}p{#1}}
\begin{table}[ht]
  \caption{Results for concrete problems: For a given problem, and any
    $\epsilon\in\EPSINT$, no adaptive priority algorithm in the
    oblivious priority function model using fewer than the specified
    number of bits of advice can achieve an approximation ratio
    smaller than the ratio listed. The last column is the number of
    advice bits required for optimality.}
\label{table:summary}
\begin{center}
\begin{tabular}{|p{0.44\linewidth}|p{0.15\linewidth}p{0.1\linewidth}|p{0.14\linewidth}|}
\hline
\raisebox{-1.5ex}{\textbf{Problem}} & \textbf{Advice for \newline Approx.} & \textbf{Approx. \newline Ratio} & \textbf{Advice for \newline Optimality} \\
\hline
&&&\\[-2ex]
Maximum Independent Set \cite{BBLP20} & $(1-H(\epsilon))n/8$ &
$1+\frac{\epsilon}{3-\epsilon}$ & $n/8$ \\[1ex]
Maximum Independent Set [Fig.~\ref{Graph3}] & $(1-H(\epsilon))n/7$ &
$1+\frac{\epsilon}{4-\epsilon}$ & $n/7$ \\[1ex]
Maximum Bipartite Matching & $(1-H(\epsilon))n/3$ &
$1+\frac{\epsilon}{3-\epsilon}$ & $n/3$ \\[1ex]
Maximum Cut & $(1-H(\epsilon))n/8$ &
$1+\frac{\epsilon}{15-\epsilon}$ & $n/8$ \\[1ex]
Minimum Vertex Cover & $(1-H(\epsilon))n/7$ &
$1+\frac{\epsilon}{3}$ & $n/7$ \\[1ex]
Maximum $3$-Satisfiability & $(1-H(\epsilon))n/3$ &
$1+\frac{\epsilon}{8-\epsilon}$ & $n/3$ \\[1ex]
Unit Job Scheduling, Prec.\ Constraints & $(1-H(\epsilon))n/9$ &
$1+\frac{\epsilon}{6-\epsilon}$ & $n/9$ \\[1ex]
\hline
\end{tabular}
\end{center}
\end{table}
Note that the gadget patterns for Maximum Independent Set from~\cite{BBLP20}
have a smaller optimal independent set than the gadget patterns for the
equivalent Minimum Vertex Cover, shown in Fig.~\ref{Graph3}. Thus, there is a
trade-off between the lower bound for the approximation ratio one can
prove and the lower bound on the amount of advice needed to prove it.

\subsection{Lower Bounds on the Advice Needed for Optimality}
\label{sec:lbopt}
In this section,
we consider adaptive priority algorithms that solve problems to optimality.

\begin{theorem}
Consider a collection of $k\geq 1$ gadget pattern pairs
$\SETOF{(G_j^a, G_j^r)}{1\leq j\leq k}$ for a problem~$B$,
satisfying the conditions in
Theorem~\ref{thm:template-minimization}.
Let $s$ be the maximum number of input items in any gadget pattern in
the collection.  Then, any optimal adaptive priority algorithm, \ALG,
with advice in the oblivious priority function model must use at least
$\FLOOR{n/s}$ advice bits on worst case instances with $n$ input
items.
\end{theorem}
\begin{proof}
  We use the proof of Theorems~\ref{thm:template-minimization}
  and~\ref{thm:template-maximization}.  Note that the reduction
  algorithm in Fig.~\ref{alg:template} uses the same amount of advice
  for the algorithm for $2$-SGKH as for the algorithm for problem~$B$
  and makes exactly the same number of errors in guessing bits for
  $2$-SGKH as it makes on first input items of gadgets for
  problem~$B$. Thus, if it solves $B$ to optimality, it also solves
  $2$-SGKH to optimality. Since $n'$ bits of advice are required on
  $n'$-bit inputs to $2$-SGKH~\cite{BHKKSS14}, $n'$ bits of advice
  must be required for $n'$ gadgets as input to problem~$B$.  If the
  maximum gadget size is $s$, then at least $\FLOOR{n/s}$ are
  necessary to achieve optimality.
\end{proof}

In the following, we consider completable problems.  A problem $B$ is
\emph{completable} if every consistent set $S'$ of $n'<n$ input items
can be completed to a consistent set $S$ of $n$ input items in such a
way that if $C'\subseteq S'$ is not an optimal solution for $S'$,
there is no subset $C=C'\cup E$ of $S$ with $E$ a subset of additional
$n-n'$ items such that $C$ is an optimal solution for $S$.  In other
words, a problem is completable if there is a way to give the
remaining input items without giving an algorithm the opportunity to
fix an earlier non-optimal decision.  For Minimum Vertex Cover and many
other problems, for example, one can complete the set $S'$ to $S$ by
adding $n-n'$ isolated vertices.

The result in Section~\ref{sec:firstexample} for Vertex Cover in the
decision-based priority function model can be generalized to give the
same result as above.
\begin{theorem}
Consider a collection of $k\geq 1$ gadget pattern pairs
$\SETOF{(G_j^a, G_j^r)}{1\leq j\leq k}$ for a completable problem~$B$,
satisfying the conditions in
Theorem~\ref{thm:template-minimization}.
Let $s$ be the maximum number of input items in any gadget pattern in
the collection.
Then, any optimal adaptive priority algorithm, \ALG, with advice in the
decision-based priority function model
must use at least $\FLOOR{n/s}$ advice bits on worst case instances
with $n$ input items.
\end{theorem}
\begin{proof}
To define a problem where $k=\FLOOR{n/s}$ bits of advice are necessary and
sufficient for optimality in the decision-based priority function model,
we consider an arbitrary algorithm, \ALGp, for problem $B$,
and an adversary, \ADVp.  We create $k$ disjoint universes,
$\mathcal{U}_1,\mathcal{U}_2,\ldots,\mathcal{U}_k$, copies of the universe~$\mathcal{U}$, with different
names for the input items in each copy, and define the universe, $\mathcal{U}'$,
for \ALGp to be the union of these $k$ universes.  The input for \ALGp
is the union of $H_1,H_2,\ldots,H_k$, where $H_i$ is an isomorphic
copy of either $G_i^a$ or $G_i^r$.

We now define $2^k$ distinct sequences of input items for \ALGp, by
describing how one of these $2^k$ sequences of input items is defined:
\ALGp selects input items one at a time, and \ADVp knows from which of
the $k$ universes the input items originate.

Since we are assuming that \ALGp solves the problem to optimality, the
adversary can assume that the current priority function is
determined based on \ALGp making the correct accept/reject decisions
up to this point. Now, \ADVp does the following: Assume that \ALGp has
already received input items originating from $i$ of the universes
from which $\mathcal{U}'$ was defined and the adversary has a current subset
$X\subseteq U'$.  If that is the case, then $X$
contains exactly enough input items to complete one gadget
from each of the universes from which \ALGp has received some input
item (how this is maintained is explained below).  From universes not
included in these $i$ universes, $X$ still contains
all possible namings of vertices from the gadgets.

Now, \ALGp receives its next input item which will be the input item in $X$
of the highest priority in this round, and that input item is the next
in the input sequence we are defining. This item is determined by the
current priority function which only depends on the input items
received so far and its decisions so far.

If that next input item, $v$, is from one of the $i$ universes,
nothing further is done. However, if that next input item originates
from a universe, $\mathcal{U}_j$, not among the $i$, then the following is done.

By the first item condition and the disjoint copies condition, the
input item $v$ identifies for which gadget pattern pair,
$(G_j^a,G_j^r)$, in the collection $v$ is the distinguishing item,
\ADVp chooses 
$G_j^a$ or $G_j^r$, 
and then removes from $X$ all input items originating from
$\mathcal{U}_j$, except enough to make up exactly the gadget that was
chosen (consistent with whichever of $G_j^a$ or $G_j^r$ was chosen),
with the naming consistent with $v$ being the distinguishing item from
that gadget.

Continuing this inductively defines one sequence of the
$2^k$ distinct sequences of input items. The number of input items in
each sequence is at most $sk\leq n$. If it is less than $n$, irrelevant
input items can be added, since $B$ is completable.

If a priority algorithm with advice for problem $B$ in the decision-based
model uses fewer than
$k$ bits of advice for instances with $sk$ input items, the same
advice must be given for at least two of the sequences, $I_1$
and~$I_2$, defined above. \ALGp therefore uses the same priorities and
makes the same decisions on $I_1$ and $I_2$ until some difference is
detected.  Thus, consider the first time in the processing of $I_1$
and~$I_2$, where an input item $v$ that has current highest priority is
the first input item of a gadget from some~$\mathcal{U}_i$, but the gadgets
included in $I_1$ and~$I_2$ from~$\mathcal{U}_i$ are different.

Up until (and including) this point, all input items have been the
same for the two sets. Thus, \ALGp must make the same decision for~$v$
in both $I_1$ and~$I_2$, but, by the distinguishing decision condition,
one of those decisions leads to a solution
which is not optimal, by the additivity of the objective function.
Thus, \ALGp is not optimal, and $k$~bits of advice
are necessary.
\end{proof}

The templates from the theorems in this section are quite similar and general,
applying to binary decision problems
where collections of gadget pattern pairs satisfying the required
conditions can be created. One can check that all of the gadget
pattern pairs presented in~\cite{BBLP20} are appropriate, thus giving
immediate lower bounds for several problems.

Recall that an exact algorithm created in the obvious way (trying all
advice strings of the maximal required length) from adaptive priority
algorithms with advice is called a \emph{priority exact algorithm}.
For any problem satisfying the conditions of the previous theorem, any
priority exact algorithm obtained for the problem examines at least
$2^{n/s}$ possibilities. This can rule out the possibility of
improvements using priority exact algorithms for certain problems that
already have known complexities better than this. When the size of the
gadget patterns is small, this gives larger lower bounds. For example,
for Minimum Vertex Cover the size of the gadget patterns is $s=7$,
since all possible gadgets have seven vertices.  Thus, the lower bound
for Minimum Vertex Cover (on triangle-free graphs with maximum
degree~$3$) is $\Omega(2^{\frac{n}{7}})\subset \Omega(1.142^n)$, which
is larger than the best known exact algorithms for this problem,
showing that those algorithms are not priority exact algorithms
(derived from a priority algorithm with advice in the decision-based
or oblivious priority function models). For Maximum Independent Set,
our previous gadget patterns~\cite{BBLP20} have size $8$, but the
gadget patterns for Minimum Vertex Cover also work for Maximum
Independent Set (the problems are complements of each other), so the
lower bound we obtain for Maximum Independent Set is also
$\Omega(2^{\frac{n}{7}})$.

Unfortunately, these lower bound results only apply to priority exact
algorithms as defined from priority algorithms with advice in either
the decision-based or the oblivious priority function models (obtained
from a priority algorithm with advice by running the algorithm on all
possible advice strings, all of the same length). As mentioned
earlier, there are usually better implementations of these algorithms
as branch-and-reduce algorithms, giving the possibility of better
analyses of their running times.

In particular, these lower bounds were all proven using constant-sized
gadget patterns, each one being a connected component of the entire
graph.  In practice, though, each connected component (gadget) should be
treated independently, each only requiring one bit of advice.
Then, if a lower bound of $f(n)$ is proven on
the number of advice bits needed for a problem of size~$n$, consisting
of $s$ components, instead of
running time $\Omega(2^{f(n)})$, only $O^*(2s)=O^*(1)$ time is
necessary (trying the advice strings ``0'' and ``1'' for each component).
Thus, it seems very limited how broadly these lower bounds
can be interpreted.

Brahe~\cite{B21} has a construction for Maximum Independent
Set and Minimum Vertex Cover 
using a connected graph which also gives
a linear lower bound on the amount of advice required for optimality
in the decision-based priority function model
(those specific connected graphs were explicitly designed to have
triangles, so they are not triangle-free, but they still have maximum
degree~3).  Thus, the technique of
running the algorithm independently on each connected component fails
there, and one obtains an exponential lower bound for exact algorithms
based on the adaptive priority priority algorithms with advice in
the decision-based priority function model.

\section{The Thorny Path Problem}
\label{sec:tp}

In this section, we consider another problem using adaptive priority
algorithms with advice. Using different techniques, we prove lower
bounds for this problem in the unrestricted and decision-based models.
We conjecture that the lower bound in the unrestricted model is not tight.
We prove
matching upper and lower bounds in the decision-based priority
function model.

We call a tree a \emph{thorny path} if it has a root, $s$, with two
children, and at any depth greater than zero and smaller than the
maximum depth of the tree, there are exactly two nodes; one with zero and one
with two children.

We define the \emph{thorny path problem} as follows. Given a
forest~$G$ consisting of a number of trees, each of which is a thorny
path, as well as a start vertex $s$ of one of the thorny paths of $G$,
the goal is to construct a path from $s$ to one of the two leaves of
maximum depth. The universe of input items is
$\mathcal{U} = \mathbb{Z}^3$. An input item $(u, v, w)$ is a vertex
$u$ with a left child $v$ and a right child $w$.  One can think of
$u$, $v$, and $w$ as vertex names or object identifiers.  The
universe of decisions is $\mathcal{D} = \{0, 1, \bot\}$. Given an
input item $(u, v, w)$, the decision $0$ means to include edge $(u,v)$
in the solution, the decision $1$ means to include edge $(u, w)$ in
the solution, and the decision $\bot$ means to not include any of the
two edges in the solution. The thorny path problem is parameterized by
a single parameter $k \in \mathbb{N}$, which is one less than the
maximum depth in the thorny path containing~$s$.  We refer to the
parameterized thorny path problem as the $k$-thorny path problem.
An example of a thorny path is shown in Fig.~\ref{fig:thorny-path}.
\begin{figure}[h]
\centerline{\includegraphics[scale=0.8]{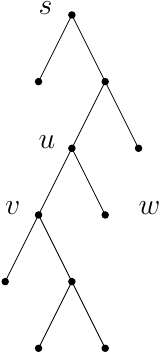}}
  \caption{An example of a $4$-thorny path.}
  \label{fig:thorny-path}
\end{figure}

We begin with a simple observation.

\begin{lemma}
  \label{thornypathupper}
  In the decision-based priority function model, the $k$-thorny path
  problem can be solved by an adaptive priority algorithm with $k$
  bits of advice.
\end{lemma}
\begin{proof}
  The first priority function gives highest priority to an input item
  of the form $(s, \cdot, \cdot)$ and an advice bit is used to select
  the correct child.  Subsequent priority functions give highest
  priority to items with the most recently selected child as the first
  entry and an advice bit is used to choose the next child correctly.
  No advice is necessary at depth~$k$, since including either edge
  gives a valid solution, a leaf at depth~$k+1$.
\end{proof}

Now we turn to lower bounds, starting with the unrestricted priority
function model. We do not give upper bounds. Note, however, that
advice giving the name of a leaf in the thorny path can be used to
follow parents up to the root, without using additional advice. This
advice can be quite large, however, since the universe size is
unbounded.

\begin{theorem}
  In the unrestricted priority function model, the $k$-thorny path
  problem cannot be solved by an adaptive priority algorithm with
  $\log k - 1$ bits of advice.
\end{theorem}
\begin{proof}
Assume that we have $\ell$ adaptive priority algorithms without
advice, $\ALG_1, \ldots, \ALG_\ell$. We fix $m$ large enough (to be
specified later) and let
$x_1, \ldots, x_m \in \mathbb{Z}\setminus\{1\}$ be distinct. Let
$\mathcal{U}$ be the input universe, consisting of all triples with
distinct items formed from $\{s,x_1, \ldots, x_m\}$, with the only
exception being that $s$ only appears as a first element of any
triple.  We construct a thorny path instance $\mathcal{I}$ (that is,
a subset of the input universe that will be used as input) with one
thorny path such that each algorithm $\ALG_1, \ldots, \ALG_\ell$
makes a mistake on $\mathcal{I}$. We construct $\mathcal{I}$
iteratively. In step~$j$, we construct a subinstance $\mathcal{I}_j$
that guarantees that algorithm $\ALG_j$ makes a mistake. The thorny
path of $\mathcal{I}_j$ starts at vertex~$s$ and ends in two
leaves. In addition to $\mathcal{I}_j$, we keep track of a leaf
$v_j$ that is going to be extended in step~$j+1$. We also keep track
of a set of input items $S_j \subseteq S$ that can be used to extend
our instance beyond $\mathcal{I}_j$.  $S_j$ will not contain any
input items where the first entry is currently a non-leaf element of
$\mathcal{I}_j$.  The condition that $\ALG_j$ makes a mistake on
$\mathcal{I}_j$ also continues to hold no matter how $\mathcal{I}_j$
is extended with elements from $S_j$.

For the base case, $\mathcal{I}_0$ is empty, and none of the
algorithms have made a mistake yet. We set $v_0 = s$ and
$S_0 = \mathcal{U}$.

Assume that we have constructed a thorny path $\mathcal{I}_j$ and the
leaf of $\mathcal{I}_j$ to be extended using items from $S_j$ is
$v_j$. Moreover each of $\ALG_1, \ldots, \ALG_j$ makes a mistake on
$\mathcal{I}_j$ and continues to make that mistake no matter how
$\mathcal{I}_j$ is extended by elements from $S_j$. Consider running
$\ALG_{j+1}$ on input $\mathcal{I}_j \cup S_j$ (in spite of it being
an invalid input). In each iteration, the algorithm gives highest
priority to an input item from $\mathcal{I}_j$ or from $S_j$. Consider
the first time $\ALG_{j+1}$ selects an input item from $S_j$.

If $\ALG_{j+1}$ has already made a mistake on an input item from
$\mathcal{I}_j$, then we can simply take
$\mathcal{I}_{j+1} = \mathcal{I}_j$, $v_{j+1} = v_j$, and
$S_{j+1} = S_j$. All the properties are easy to verify in this case.

Otherwise, let $(x,y,z)$ be the first element from $S_j$ that is
requested by $\ALG_{j+1}$. Without loss of generality, assume that the
decision of $\ALG_{j+1}$ is to accept edge $(x,y)$ and not $(x,z)$. If
$x = v_j$, then we extend
$\mathcal{I}_{j+1} = \mathcal{I}_{j} \cup \{(x,y,z)\}$ and $S_{j+1}$
is $S_j$ with all items involving $y$ or $x$ removed, as well as those
items that have $z$ as second or third coordinate. Observe that this
ensures that $\ALG_{j+1}$ makes a mistake on item $(x,y,z)$ and this
fact is unaffected by further extensions of $\mathcal{I}_{j+1}$. In
this case, we have $v_{j+1} = z$.

The last case to consider is when $\ALG_{j+1}$ requests $(x,y,z)$ from
$S_j$ and $x \ne v_j$. In this case, we also consider an item
$(v_j, x, w) \in S_j$ for some $w$ that is different from any other
value appearing in the construction so far. By the way $S_j$ is
constructed, and taking $m$ large enough, such a $w$ is guaranteed to
exist. We extend
$\mathcal{I}_{j+1} = \mathcal{I}_j \cup \{(v_j, x, w), (x, y,
z)\}$. Again, without loss of generality, assume that $\ALG_{j+1}$
accepts $(x,y)$ rather than $(x,z)$. We again set $v_{j+1} = z$ and
$S_{j+1}$ to be the set $S_j$ with all items involving $x$, $y$, $w$,
or $v_j$ removed, as well as those items that have $z$ as the second
or third coordinate. This guarantees that $\ALG_{j+1}$ makes a mistake
on item $(x,y,z)$ and continues to make a mistake on this item no
matter how $\mathcal{I}_{j+1}$ is extended with elements from
$S_{j+1}$.

After all $\ell$ algorithms have made a mistake, leaving a final $v_j$
and $S_j$, an input item from $S_j$ with $v_j$ as the first coordinate
is moved from $S_j$ to $\mathcal{I}_j$, finishing the construction.

Now, observe that each $S_j$ can be defined by some subset
$F \subseteq \{s, x_1, \ldots, x_m\}$. Namely, $S_j$ consists of all
triples formed from $F$, as well as triples formed by having the first
coordinate equal to $v_j$ and the remaining two coordinates coming
from~$F$. In each iteration going from $j$ to $j+1$, at most $4$
elements are removed from~$F$. At the end, three additional elements
from $F$ are used for the last item. Therefore, $m = 3 + 4 \ell$ is
sufficient to guarantee a universe large enough that the construction
terminates only after all algorithms are fooled by the instance.

Finally, assume that $b$ advice bits are used by an adaptive priority
algorithm with advice with the above construction as input.  We
determine a lower bound on~$b$.  Running an algorithm in the
unrestricted priority function with $b$ bits of advice is equivalent
to running $2^b$ algorithms in parallel. Thus, we have $\ell=2^b$
algorithms that can all be fooled simultaneously by a $k$-thorny path
problem, where $k\leq 2\ell$, since the last case above uses two
layers to fool the algorithm in question.  Since $b$ bits are
insufficient and $2^{b+1}=2\ell\geq k$, it follows that $\log k - 1$
bits are insufficient.
\end{proof}

The following theorem shows that the upper bound
in Lemma~\ref{thornypathupper} is tight for the decision-based priority
function model. The proof uses the same ideas as the proof of
the lower bound in Section~\ref{sec:firstexample}. In that proof,
$2^k$ different input sequences were created, and using fewer than
$k$ bits of advice led to at least two of them getting the same advice
and an error being made on one of those two. Those sequences can be
seen as forming a binary tree, with inputs at the nodes in the tree
and the two possible decisions leading to the two subtrees. Thus,
sequences that are the same up until input $m$ share the same path
from the root to that input. This is not quite the case for the
thorny path problem, since it is possible for the adaptive priority
algorithm to select an input item that is not connected to the last
one seen. However, the tree determines $2^k$ root to leaf paths, which
naturally define $2^k$ thorny paths and their inputs.

\begin{theorem}
  An adaptive priority algorithm with advice in the decision-based
  priority function model must use at least $k$ bits of advice to
  solve the $k$-thorny path problem.
\end{theorem}

\begin{proof}
Let \ALG be an adaptive priority algorithm with advice in the
decision-based priority function model. We consider \ALG's
computation on the $k$-thorny path problem.

We construct $2^{k}$ input instances of the $k$-thorny path problem.
To explain the construction, we use a binary tree with $2^{k+1}$
leaves and $s$ as the root. The leaves will be the leaves in the
thorny path problems, and each root to leaf path, along with the
siblings of the vertices on the path, will be the thorny paths that
should be followed by \ALG to get to a leaf.  The $2^{k}$ different
paths one can take from~$s$ to a parent of a leaf will represent the
$2^{k}$ input instances we are constructing.  They will not \emph{be}
the input instances since input instances could have further input
items that are discarded by \ALG.  Each node in the tree has an
associated ordered list of input items, which are all the ones for
which the algorithm chooses $\bot$ (discard) until the next time it
chooses $0$ or~$1$ (left or right).  Thus, a path in the tree defines
an input sequence consisting of the input items forming the path with
the associated ordered lists of input items added. More precisely, the
ordered list of input items associated with a node~$u$ appears in the
input sequence just prior to the input item with $u$ as root (of that
input item; recall that an input item consists of three nodes, two of
which are the children of the first, the root).

The tree represents all execution paths \ALG can take based on
different advice. Along the way, we will also explain how the
adversary will change the input universe as an execution proceeds.  In
an execution (that follows one path), the universe is decreased
gradually as execution progresses down the path, and the universe that
is used at a given point varies depending on which path was chosen by
\ALG (based on its advice).

In constructing the tree, we start with $s$ in the root and we add
nodes to the tree gradually by adding two children to a currently
childless node. Let $u$ be such a node of depth at most~$k$.  We
consider \ALG's execution on the partial input defined by the path
from $s$ to $u$ (including the input items associated with nodes on
the path). The path, together with the associated lists, defines the
decisions \ALG must make on this partial input.

Naturally, \ALG just follows one path in the tree, making decisions to
go left or right or discard based on the advice it gets.  However, if
it is at~$u$, then the next input item $\ALG$'s priority function
selects is only based on the partial input and its decisions.  Now,
for an input item, $(x,y,z)$, either $x=u$ or $x$ is not on the $s$ to
$u$ path (this follows from how we treat the universe; see later).

If $x=u$, we add the leaves $y$ and $z$ to the tree as children
of~$u$.  All input items remaining containing $u$ or its sibling are
removed from the universe.

If $x\not=u$, the adversary removes all input items remaining that
contain $x$, $y$, or $z$ from the universe.  Thus, $x$ can never
become part of any root to leaf path that currently ends at~$u$.  The
input item $(x,y,z)$ is then appended to the ordered list of discarded
input items associated with~$u$.

There are no more input items added after there are $2^{k+1}$ leaves
at depth~$k+1$. If the tree is never completed, there is a path where
\ALG never finds a leaf, so \ALG fails. Otherwise, the $2^{k}$
different input sequences defined by the paths in the tree must have
each their distinct advice string.  Thus, at least $k$ advice bits are
necessary.

\end{proof}

Note that this proof does not appear to work in the unrestricted
priority function model, since it is not clear that the tree can be
defined. For example, if advice (in addition to the decisions made) is
used to determine which input item is chosen next, an input that we
placed off of a thorny path might actually only be chosen if it is on
the path.

\section{Open Problems}
\label{sec:conclusion}

The extension of the adaptive priority model to the advice tape model
leads to many new research directions. We consider the following open
problems to be of particular interest:

\begin{itemize}
\item Design and analyze new adaptive priority algorithms with advice
  for (special cases of) classical optimization problems and convert
  them to offline algorithms, by trying all possibilities for the
  advice as with priority exact algorithms or by implementing them as
  branch-and-reduce algorithms. In particular, are there priority
  algorithms with advice that lead to faster (in terms of the base of
  the exponent) exact exponential time offline algorithms than the
  best known?
\item The previous question also applies to approximation algorithms,
  when the best known offline approximation algorithm is exponential
  in terms of running time.
\item Suggest how to extend the lower bound results to the
  unrestricted priority function model.
  A first example of such a lower bound for an artificial problem
  was given in Section~\ref{sec:tp} for the thorny path problem.
\item Suggest and investigate other extensions of the adaptive
  priority framework besides the information-theoretic advice tape
  extension. For instance, one could consider a class of adaptive
  priority algorithms where advice is given by an $\mbox{AC}^0$
  circuit. What can be said about the power and limitations of such
  algorithms?
\item More generally, study the structural complexity of priority
  algorithms with advice. What reasonable complexity classes can be
  defined based on advice complexity and approximation ratio?
\item The lower bounds implied by our reduction-based framework are of
  the form ``constant inapproximability even given linear advice.''
  Can this framework be extended to handle super-constant
  inapproximability with sublinear advice? More generally, the goal is
  to design some framework that could work in this other realm of
  parameters. A good starting point would be to show that Maximum
  Independent Set cannot be approximated to within $n^{1-\epsilon}$
  with $O(\log n)$ bits of advice, for any fixed $\epsilon \in (0,
  1]$. Note that under the assumption $\textrm{P}\neq\textrm{NP}$,
  this lower bound follows from the famous result of
  H{\aa{}}stad~\cite{Hastad99}. The goal here is to prove this lower
  bound unconditionally for the restricted class of priority
  algorithms with advice.
\end{itemize}

\section*{Acknowledgements}

We would like to thank Nicolai Bille Brahe for pointing out an
ambiguity in an earlier version of this paper.
%and Thomas B\"{u}low Buhl for pointing out typos.

The first and second authors were supported in part by the Independent Research Fund Denmark, Natural Sciences, grants DFF-7014-00041 and DFF-0135-00018B, and the third author was supported in part by Natural Sciences and Engineering Research Council (NSERC) of Canada.

\bibliographystyle{plain}
\bibliography{refs}

\begin{thebibliography}{10}

\bibitem{ABM04}
Dimitris Achlioptas, Paul Beame, and Michael Molloy.
\newblock A sharp threshold in proof complexity yields lower bounds for
  satisfiability search.
\newblock {\em Journal of Computer and System Sciences}, 68(2):238--268, 2004.

\bibitem{alekhnovich2005toward}
Michael Alekhnovich, Allan Borodin, Joshua Buresh{-}Oppenheim, Russell
  Impagliazzo, Avner Magen, and Toniann Pitassi.
\newblock Toward a model for backtracking and dynamic programming.
\newblock {\em Computational Complexity}, 20(4):679--740, 2011.

\bibitem{AHI05}
Michael Alekhnovich, Edward~A. Hirsch, and Dmitry Itsykson.
\newblock Exponential lower bounds for the running time of {DPLL} algorithms on
  satisfiable formulas.
\newblock {\em Journal of Automated Reasoning}, 35(1--3):51--72, 2005.

\bibitem{AngelopoulosB2004}
Spyros Angelopoulos and Allan Borodin.
\newblock On the power of priority algorithms for facility location and set
  cover.
\newblock {\em Algorithmica}, 40(4):271--291, 2004.

\bibitem{BBKTW94}
Shai Ben{-}David, Allan Borodin, Richard~M. Karp, G{\'{a}}bor Tardos, and Avi
  Wigderson.
\newblock On the power of randomization in on-line algorithms.
\newblock {\em Algorithmica}, 11(1):2--14, 1994.

\bibitem{BesserP17}
Bert Besser and Matthias Poloczek.
\newblock Greedy matching: Guarantees and limitations.
\newblock {\em Algorithmica}, 77(1):201--234, 2017.

\bibitem{DBLP:journals/ijfcs/BianchiBBKP18}
Maria~Paola Bianchi, Hans{-}Joachim B{\"{o}}ckenhauer, Tatjana
  Br{\"{u}}lisauer, Dennis Komm, and Beatrice Palano.
\newblock Online minimum spanning tree with advice.
\newblock {\em International Journal of Foundations of Computer Science},
  29(4):505--527, 2018.

\bibitem{BHKKSS14}
Hans-Joachim B{\"o}ckenhauer, Juraj Hromkovi{\v{c}}, Dennis Komm, Sacha Krug,
  Jasmin Smula, and Andreas Sprock.
\newblock The string guessing problem as a method to prove lower bounds on the
  advice complexity.
\newblock {\em Theoretical Computer Science}, 554:95--108, 2014.

\bibitem{BKKKM17}
Hans-Joachim B\"{o}ckenhauer, Dennis Komm, Rastislav Kr\'{a}lovi\v{c}, Richard
  Kr\'{a}lovi\v{c}, and Tobias M\"{o}mke.
\newblock Online algorithms with advice: The tape model.
\newblock {\em Information and Computation}, 254(1):59--83, 2017.

\bibitem{BorodinBLM2010}
Allan Borodin, Joan Boyar, Kim~S. Larsen, and Nazanin Mirmohammadi.
\newblock Priority algorithms for graph optimization problems.
\newblock {\em Theoretical Computer Science}, 411(1):239--258, 2010.

\bibitem{BBLP20}
Allan Borodin, Joan Boyar, Kim~S. Larsen, and Denis Pankratov.
\newblock Advice complexity of priority algorithms.
\newblock {\em Theory of Computing Systems}, 64:593--625, 2020.

\bibitem{BorodinL16}
Allan Borodin and Brendan Lucier.
\newblock On the limitations of greedy mechanism design for truthful
  combinatorial auctions.
\newblock {\em {ACM} Transactions on Economics and Computation},
  5(1):2:1--2:23, 2016.

\bibitem{BorodinNR2003}
Allan Borodin, Morten~N. Nielsen, and Charles Rackoff.
\newblock ({Incremental}) priority algorithms.
\newblock {\em Algorithmica}, 37(4):295--326, 2003.

\bibitem{BoyarFKLM2016}
Joan Boyar, Lene~M. Favrholdt, Christian Kudahl, Kim~S. Larsen, and Jesper~W.
  Mikkelsen.
\newblock {Online Algorithms with Advice: A Survey}.
\newblock {\em ACM Computing Surveys}, 50(2):19:1--19:34, 2017.

\bibitem{B21}
Nicolai~Bille Brahe.
\newblock Exact algorithms from priority algorithms with advice for vertex
  cover in graphs of maximum degree 3.
\newblock Master's thesis, University of Southern Denmark, Denmark, 2021.

\bibitem{C77}
Vasek Chv{\'{a}}tal.
\newblock Determining the stability number of a graph.
\newblock {\em {SIAM} Journal on Computing}, 6(4):643--662, 1977.

\bibitem{C80}
Vasek Chvátal.
\newblock Hard knapsack problems.
\newblock {\em Operations Research}, 28(6):1402--1411, 1980.

\bibitem{DLL62}
Martin Davis, George Logemann, and Donald~W. Loveland.
\newblock A machine program for theorem-proving.
\newblock {\em Commun. {ACM}}, 5(7):394--397, 1962.

\bibitem{DP60}
Martin Davis and Hilary Putnam.
\newblock A computing procedure for quantification theory.
\newblock {\em Journal of the {ACM}}, 7(3):201--215, 1960.

\bibitem{DavisI2009}
Sashka Davis and Russell Impagliazzo.
\newblock Models of greedy algorithms for graph problems.
\newblock {\em Algorithmica}, 54(3):269--317, 2009.

\bibitem{Ais}
Stefan Dobrev, Rastislav Kr\'alovi\v{c}, and Richard Kr\'alovi\v{c}.
\newblock Advice complexity of maximum independent set in sparse and bipartite
  graphs.
\newblock {\em Theoretical Computer Science}, 56(1):197--219, 2015.

\bibitem{DBLP:conf/sirocco/DobrevKM12}
Stefan Dobrev, Rastislav Kr\'alovi\v{c}, and Euripides Markou.
\newblock Online graph exploration with advice.
\newblock In {\em 19th International Colloquium on Structural Information and
  Communication Complexity (SIROCCO)}, volume 7355 of {\em Lecture Notes in
  Computer Science}, pages 267--278, Berlin, Heidelberg, Germany, 2012.
  Springer.

\bibitem{A4}
Stefan Dobrev, Rastislav Kr\'alovi\v{c}, and Dana Pardubsk{\'a}.
\newblock Measuring the problem-relevant information in input.
\newblock {\em {RAIRO} -- Theoretical Informatics and Applications},
  43(3):585--613, 2009.

\bibitem{A2}
Yuval Emek, Pierre Fraigniaud, Amos Korman, and Adi Ros{\'e}n.
\newblock Online computation with advice.
\newblock {\em Theoretical Computer Science}, 412(24):2642--2656, 2011.

\bibitem{DBLP:journals/jacm/FominGK09}
Fedor~V. Fomin, Fabrizio Grandoni, and Dieter Kratsch.
\newblock A measure {\&} conquer approach for the analysis of exact algorithms.
\newblock {\em Journal of the {ACM}}, 56(5):25:1--25:32, 2009.

\bibitem{FK10}
Fedor~V. Fomin and Dieter Kratsch.
\newblock {\em Exact Exponential Algorithms}.
\newblock Texts in Theoretical Computer Science. An EATCS Series. Springer,
  Berlin, Heidelberg, Germany, 2010.

\bibitem{DBLP:journals/iandc/FraigniaudIP08}
Pierre Fraigniaud, David Ilcinkas, and Andrzej Pelc.
\newblock Tree exploration with advice.
\newblock {\em Information and Computation}, 206(11):1276--1287, 2008.

\bibitem{GP19}
Barun Gorain and Andrzej Pelc.
\newblock Deterministic graph exploration with advice.
\newblock {\em {ACM} Transactions on Algorithms}, 15(1):8:1--8:17, 2019.

\bibitem{Graham1966}
Ronald~L. Graham.
\newblock Bounds for certain multiprocessing anomalies.
\newblock {\em Bell Systems Technical Journal}, 45(9):1563--1581, 1966.

\bibitem{Magnus}
Magn{\'u}s~M. Halld{\'o}rsson, Kazuo Iwama, Shuichi Miyazaki, and Shiro
  Taketomi.
\newblock Online independent sets.
\newblock {\em Theoretical Computer Science}, 289(2):953--962, 2002.

\bibitem{Hastad99}
Johan H{\aa}stad.
\newblock Clique is hard to approximate within $n^{1-\epsilon}$.
\newblock {\em Acta Mathematica}, 182(1):105--142, 1999.

\bibitem{A3}
Juraj Hromkovi{\v{c}}, Rastislav Kr\'alovi\v{c}, and Richard Kr\'alovi\v{c}.
\newblock Information complexity of online problems.
\newblock In {\em 35th International Symposium on Mathematical Foundations of
  Computer Science (MFCS)}, volume 6281 of {\em Lecture Notes in Computer
  Science}, pages 24--36, Berlin, Heidelberg, Germany, 2010. Springer.

\bibitem{KargerSW99}
David~R. Karger, Cliff Stein, and Joel Wein.
\newblock Scheduling algorithms.
\newblock In Mikhail~J. Atallah, editor, {\em Algorithms and Theory of
  Computation Handbook}, Chapman {\&} Hall/CRC Applied Algorithms and Data
  Structures series. {CRC} Press, Boca Raton, Florida, USA, 1999.

\bibitem{DBLP:journals/corr/KommKKK15}
Dennis Komm, Rastislav Kr\'alovi\v{c}, Richard Kr\'alovi\v{c}, and Christian
  Kudahl.
\newblock Advice complexity of the online induced subgraph problem.
\newblock In {\em 41st International Symposium on Mathematical Foundations of
  Computer Science (MFCS)}, volume~58 of {\em LIPIcs}, pages 59:1--59:13,
  Germany, 2016. Schloss Dagstuhl - Leibniz-Zentrum fuer Informatik.

\bibitem{Ahunt}
Dennis Komm, Rastislav Kr\'alovi\v{c}, Richard Kr\'alovi\v{c}, and Jasmin
  Smula.
\newblock Treasure hunt with advice.
\newblock In {\em 22nd International Colloquium on Structural Information and
  Communication Complexity (SIROCCO)}, volume 9439 of {\em Lecture Notes in
  Computer Science}, pages 328--341, Berlin, Heidelberg, Germany, 2015.
  Springer.

\bibitem{KorteL19842}
B.~Korte and L.~Lov{\'{a}}sz.
\newblock Greedoids and linear objective functions.
\newblock {\em SIAM Journal on Algebraic Discrete Methods}, 5(2):229--238,
  1984.

\bibitem{KorteL1981}
Bernhard Korte and L\'{a}szl\'{o} Lov\'{a}sz.
\newblock Mathematical structures underlying greedy algorithms.
\newblock In {\em International {FCT}-Conference on Fundamentals of Computation
  Theory (FCT)}, pages 205--209, Berlin, Heidelberg, Germany, 1981. Springer.

\bibitem{KorteL1983}
Bernhard Korte and L{\'a}szl{\'o} Lov{\'a}sz.
\newblock Structural properties of greedoids.
\newblock {\em Combinatorica}, 3(3):359--374, 1983.

\bibitem{KorteL1984}
Bernhard Korte and L{\'{a}}szl{\'{o}} Lov{\'{a}}sz.
\newblock Greedoids -- a structural framework for the greedy algorithm.
\newblock In William~R. Pulleyblank, editor, {\em Progress in Combinatorial
  Optimization}, pages 221--243. Academic Press, Cambridge, Massachusetts, USA,
  1984.

\bibitem{LeshM06}
Neal Lesh and Michael Mitzenmacher.
\newblock Bubblesearch: {A} simple heuristic for improving priority-based
  greedy algorithms.
\newblock {\em Information Processing Letters}, 97(4):161--169, 2006.

\bibitem{LiptonT94}
Richard~J. Lipton and Andrew Tomkins.
\newblock Online interval scheduling.
\newblock In {\em 5th Annual {ACM-SIAM} Symposium on Discrete Algorithms
  (SODA)}, pages 302--311, 1994.

\bibitem{M79}
Colin McDiarmid.
\newblock Determining the chromatic number of a graph.
\newblock {\em {SIAM} Journal on Computing}, 8(1):1--14, 1979.

\bibitem{Abipfinal}
Shuichi Miyazaki.
\newblock On the advice complexity of online bipartite matching and online
  stable marriage.
\newblock {\em Information Processing Letters}, 114(12):714--717, 2014.

\bibitem{Papakonstantinou06}
Periklis~A. Papakonstantinou.
\newblock Hierarchies for classes of priority algorithms for job scheduling.
\newblock {\em Theoretical Computer Science}, 352(1-3):181--189, 2006.

\bibitem{Poloczek11}
Matthias Poloczek.
\newblock Bounds on greedy algorithms for {MAX} {SAT}.
\newblock In {\em 19th Annual European Symposium on Algorithms (ESA)}, volume
  6942 of {\em Lecture Notes in Computer Science}, pages 37--48, Berlin,
  Heidelberg, Germany, 2011. Springer.

\bibitem{PoloczekSWZ17}
Matthias Poloczek, Georg Schnitger, David~P. Williamson, and Anke van Zuylen.
\newblock Greedy algorithms for the maximum satisfiability problem: Simple
  algorithms and inapproximability bounds.
\newblock {\em {SIAM} Journal on Computing}, 46(3):1029--1061, 2017.

\bibitem{PI00}
Pavel Pudl{\'{a}}k and Russell Impagliazzo.
\newblock A lower bound for {DLL} algorithms for \emph{k}-sat (preliminary
  version).
\newblock In David~B. Shmoys, editor, {\em 11th Annual {ACM-SIAM} Symposium on
  Discrete Algorithms (SODA)}, pages 128--136, New York, New York/Philadelphia,
  Pennsylvania, USA, 2000. {ACM/SIAM}.

\bibitem{RaghavanS94}
Prabhakar Raghavan and Marc Snir.
\newblock Memory versus randomization in on-line algorithms.
\newblock {\em IBM Journal of Research and Development}, 38(6):683--707, 1994.

\bibitem{Regev02}
Oded Regev.
\newblock Priority algorithms for makespan minimization in the subset model.
\newblock {\em Information Processing Letters}, 84(3):153--157, 2002.

\bibitem{Whitney35}
Hassler Whitney.
\newblock On the abstract properties of linear dependence.
\newblock {\em American Journal of Mathematics}, 57(3):509--533, 1935.

\bibitem{DBLP:journals/tcs/XiaoN13}
Mingyu Xiao and Hiroshi Nagamochi.
\newblock Confining sets and avoiding bottleneck cases: {A} simple maximum
  independent set algorithm in degree-3 graphs.
\newblock {\em Theoretical Computer Science}, 469:92--104, 2013.

\end{thebibliography}

\newpage

\appendix

\section{Gadgets}
The results in Table~\ref{table:summary} are based on gadget pattern
pairs that were presented in~\cite{BBLP20}. For completeness, we
include them in this appendix.

\subsection{Maximum Independent Set and Maximum Cut}
The gadgets are drawn to have vertex $1$ be the one with highest priority.

\begin{figure}[ht]
\centering

%\begin{tikzpicture}[scale=0.35, every node/.style={scale=0.7}]
\begin{tikzpicture}[scale=0.55, every node/.style={scale=1.0}]

\node[draw=black,circle,fill=white] (4) at (0,0) {4};
\node[draw=black,circle,fill=white] (5) at (2,0) {5};
\node[draw=black,circle,fill=white] (6) at (4,0) {6};
\node[draw=black,circle,fill=white] (7) at (6,0) {7};
\node[draw=black,circle,fill=white] (8) at (8,0) {8};

\node[draw=black,circle,fill=white] (1) at (2,-4) {1};
\node[draw=black,circle,fill=white] (2) at (4,-4) {2};
\node[draw=black,circle,fill=white] (3) at (6,-4) {3};

\draw (6) -- (1);
\draw (6) -- (2);
\draw (6) -- (3);

\draw (8) -- (1);
\draw (8) -- (2);
\draw (8) -- (3);

\draw (4) -- (1);
\draw (4) -- (2);
\draw (4) -- (3);

\draw (5) -- (1);
\draw (5) -- (2);
\draw (5) -- (3);

\draw (7) -- (1);
\draw (7) -- (2);
\draw (7) -- (3);

\draw (8) -- (1);
\draw (8) -- (2);
\draw (8) -- (3);

\draw (6) -- (5);
\draw (5) -- (4);
\draw (6) -- (7);
\draw (7) -- (8);
\path (8) edge[bend right] (4);

\end{tikzpicture} 
\hspace{1cm}
%\begin{tikzpicture}[scale=0.35, every node/.style={scale=0.7}]
\begin{tikzpicture}[scale=0.55, every node/.style={scale=1.0}]

\node[draw=black,circle,fill=white] (1) at (0,0) {1};
\node[draw=black,circle,fill=white] (4) at (2,0) {4};
\node[draw=black,circle,fill=white] (3) at (4,0) {3};
\node[draw=black,circle,fill=white] (2) at (6,0) {2};
\node[draw=black,circle,fill=white] (5) at (8,0) {5};

\node[draw=black,circle,fill=white] (6) at (2,-4) {6};
\node[draw=black,circle,fill=white] (7) at (4,-4) {7};
\node[draw=black,circle,fill=white] (8) at (6,-4) {8};

\draw (1) -- (6);
\draw (1) -- (7);
\draw (1) -- (8);

\draw (2) -- (6);
\draw (2) -- (7);
\draw (2) -- (8);

\draw (3) -- (6);
\draw (3) -- (7);
\draw (3) -- (8);

\draw (4) -- (6);
\draw (4) -- (7);
\draw (4) -- (8);

\draw (5) -- (6);
\draw (5) -- (7);
\draw (5) -- (8);

\draw (1) -- (4);
\draw (2) -- (3);
\draw (3) -- (4);
\draw (2) -- (5);
\path (5) edge[bend right] (1);

\end{tikzpicture} 

\caption{Topological structure of the gadgets $(G^1,G^2)$ for independent set.}\label{fig:is-gadget}
\end{figure}

\subsubsection{Maximum Independent Set}
The optimal decision is to accept in $G^1$ and reject in $G^2$.
  The maximum number $s$ of input items for a gadget is
$8$, $\OPT(G^1)=\OPT(G^2)=3$, and $\BAD(G^1)=\BAD(G^2)=2$.

\subsubsection{Maximum Cut}
 The goal is to partition the vertices into
two sets such that the number of edges
crossing the two sets is maximized. The partition is specified by the
algorithm assigning $0$ or $1$ to each vertex. In addition, we require
that $0$ is assigned to vertices belonging to the larger block of the
partition. The maximum cut in $G^1$ (or $G^2$) puts the upper vertices
in the larger set and the lower vertices in the other set.
 The optimal decision for the first
vertex is unique: For $G^1$, respond $1$, and for $G^2$, respond $0$.
The maximum number $s$ of input items for a gadget is
$8$, $\OPT(G^1)=\OPT(G^2)=15$, and $\BAD(G^1)=\BAD(G^2)=14$.

\subsection{Maximum Bipartite Matching}
The vertices on the right-hand side are known in advance, and the
vertices on the left arrive online.
The gadgets are drawn to have vertex $1$ be the one with highest priority,
and all possible first vertices look identical.
The optimal
decision is to accept in $G^1$ and reject in $G^2$.

\begin{figure}[ht]
\centering

\begin{tikzpicture}[scale=0.5]

\node[draw=black,circle,fill=white] (l1) at (0,0) {1};
\node[draw=black,circle,fill=white] (l2) at (0,-2) {2};
\node[draw=black,circle,fill=white] (l3) at (0,-4) {3};

\node[draw=black,circle,fill=white] (r1) at (4,0) {1};
\node[draw=black,circle,fill=white] (r2) at (4,-2) {2};
\node[draw=black,circle,fill=white] (r3) at (4,-4) {3};

\draw (l1) -- (r1);
\draw (l1) -- (r2);
\draw (l2) -- (r2);
\draw (l2) -- (r3);
\draw (l3) -- (r2);
\draw (l3) -- (r3);
\end{tikzpicture} 
\hspace{1cm}
\begin{tikzpicture}[scale=0.5]

\node[draw=black,circle,fill=white] (l1) at (0,0) {1};
\node[draw=black,circle,fill=white] (l2) at (0,-2) {2};
\node[draw=black,circle,fill=white] (l3) at (0,-4) {3};

\node[draw=black,circle,fill=white] (r1) at (4,0) {1};
\node[draw=black,circle,fill=white] (r2) at (4,-2) {2};
\node[draw=black,circle,fill=white] (r3) at (4,-4) {3};

\draw (l1) -- (r1);
\draw (l1) -- (r2);
\draw (l2) -- (r1);
\draw (l2) -- (r3);
\draw (l3) -- (r1);
\draw (l3) -- (r3);
\end{tikzpicture} 

\caption{Topological structure of the gadgets $(G^1,G^2)$ for bipartite matching.}\label{fig:bm-gadget}
\end{figure}

The (maximum) number $s$ of input items (the number of
vertices given) for any of the two gadgets is $3$,
$\OPT(G^1)=\OPT(G^2)=3$, and $\BAD(G^1)=\BAD(G^2)=2$.

\subsection{Maximum Satisfiability (MAX-3-SAT)}

   An input item $(x, S^+,S^-)$ consists of a
variable name $x$, a set $S^+$ of \emph{clause information tuples} for
those clauses in which $x$ appears positively, and a set $S^-$ of
clause information tuples for those clauses where the variable $x$
appears negatively. The clause information tuples for a particular
clause contain the name of the clause, the total number of literals in
that clause, and the names of the other variables in the clause, but
no information regarding whether those other variables are negated or
not.   The
goal is to satisfy the maximum number of clauses.

\[ G^1 = C_1 \wedge C_2 \wedge C_3 \wedge C_4 \wedge C_5 \wedge C_6 \wedge C_7 \wedge C_8, \]
where
$$\begin{array}{ll}
C_1= (x_1 \vee x_2 \vee x_3) & C_2= (x_1 \vee \lnot x_2 \vee \lnot x_3) \\[1ex]
C_3= (x_1 \vee \lnot x_2 \vee x_3) & C_4= (x_1 \vee x_2 \vee \lnot x_3) \\[1ex]
C_5= (\lnot x_1 \vee x_2 \vee x_3) & C_6= (\lnot x_1 \vee x_2 \vee x_3) \\[1ex]
C_7= (\lnot x_1 \vee \lnot x_2 \vee \lnot x_3) & C_8= (\lnot x_1 \vee \lnot x_2 \vee \lnot x_3) 
\end{array}$$

\[ G^2 = C_1 \wedge C_2 \wedge C_3 \wedge C_4 \wedge C_5 \wedge C_6 \wedge C_7 \wedge C_8, \]
where
$$\begin{array}{ll}
C_1= (\lnot x_1 \vee x_2 \vee x_3) & C_2= (\lnot x_1 \vee \lnot x_2 \vee \lnot x_3) \\[1ex]
C_3= (\lnot x_1 \vee \lnot x_2 \vee x_3) & C_4= (\lnot x_1 \vee x_2 \vee \lnot x_3) \\[1ex]
C_5= (x_1 \vee x_2 \vee x_3) & C_6= (x_1 \vee x_2 \vee x_3) \\[1ex]
C_7= (x_1 \vee \lnot x_2 \vee \lnot x_3) & C_8= (x_1 \vee \lnot x_2 \vee \lnot x_3) 
\end{array}$$

Suppose without loss of generality that the highest priority input is
\[
(x_1, \{ (C_1,3,\{ x_2,x_3\}),(C_2,3,\{ x_2,x_3\}),(C_3,3,\{ x_2,x_3\}),(C_4,3,\{ x_2,x_3\})\},\]
\[\{ (C_5,3,\{ x_2,x_3\}),(C_6,3,\{ x_2,x_3\}), (C_7,3,\{ x_2,x_3\}),(C_8,3,\{ x_2,x_3\}) \}).
\]
Note that the optimal decision for $x_1$ is unique for each of these
gadgets and is ``True'' for $G^1$ and ``False'' for $G^2$.
The maximum number $s$ of input items for a gadget is
$3$, $\OPT(G^1)=\OPT(G^2)=8$, and $\BAD(G^1)=\BAD(G^2)=7$.

\subsection{Unit Job Scheduling with Precedence Constraints}

In this problem, we have a single machine and the requests
are unit time jobs with precedence constraints, indicating which jobs
must be scheduled before which others.
There could be a cyclic set of constraints. The goal is to
 schedule a maximum number of jobs that are compatible.
 The input items
are of the form $(J, S^+, S^-)$, where $J$ is the name of a job, $S^+$
is the set of jobs such that if they were scheduled together with $J$
they would have to be scheduled before $J$, and $S^-$ is the set of
jobs such that if they were scheduled together with $J$ they would
have to be scheduled after $J$.

The gadget below is a directed graph, specifying the precedence
constraints.

\begin{figure}[ht]
\centering
\tikzstyle{vertex}=[circle, draw, minimum size=5mm] % , inner sep=0pt
\newcommand{\vertex}{\node[vertex]}
\usetikzlibrary{arrows.meta}
\tikzset{%
  tipA/.tip={Triangle[open,angle=75:3pt]},
  tipB/.tip={Triangle[angle=75:5pt]},
}
\begin{tikzpicture}[scale=1.0] % ,line width=0.75
\vertex (p0) at (0,2) {0};
\vertex (p1) at (2,2) {1};
\vertex (p2) at (4,2) {2};
\vertex (p3) at (6,2) {3};
\vertex (p4) at (0,0) {4};
\vertex (p5) at (2,0) {5};
\vertex (p6) at (4,0) {6};
\vertex (p7) at (6,0) {7};
\vertex (p8) at (1,1) {8};
\draw[-tipB] (p0) -- (p8);
\draw[-tipB] (p0) to[out=45,in=135,distance=7mm] (p2);
\draw[-tipB] (p1) -- (p0);
\draw[-tipB] (p1) -- (p2);
\draw[-tipB] (p2) -- (p3);
\draw[-tipB] (p2) -- (p7);
\draw[-tipB] (p3) to[out=90,in=90,distance=10mm] (p0);
\draw[-tipB] (p3) to[out=135,in=45,distance=7mm] (p1);
\draw[-tipB] (p4) to[out=-45,in=-135,distance=7mm] (p6);
\draw[-tipB] (p4) -- (p8);
\draw[-tipB] (p5) -- (p4);
\draw[-tipB] (p5) -- (p6);
\draw[-tipB] (p6) -- (p3);
\draw[-tipB] (p6) -- (p7);
\draw[-tipB] (p7) to[out=-90,in=-90,distance=10mm] (p4);
\draw[-tipB] (p7) to[out=-135,in=-45,distance=7mm] (p5);
\draw[-tipB] (p8) -- (p1);
\draw[-tipB] (p8) -- (p5);
\end{tikzpicture}
\caption{Topological structure of a gadget for job scheduling of unit
  time jobs with precedence constraints.}\label{fig:js-gadget}
\end{figure}
This gadget  consists only of isomorphic items (each vertex
has in-degree $2$, out-degree $2$, and $4$ different neighbors in
all). Thus, this gadget can represent both $G^1$ and $G^2$ with
renaming.
Every optimal solution contains Job~0 and
excludes Job~8, so $G^1$ has the job labeled 0 in this gadget as the highest priority item and
$G^2$ has the job labeled 8 in this gadget as the highest priority item.
The maximum number $s$ of input items for a gadget is
$9$, $\OPT(G^1)=\OPT(G^2)=6$ (for instance, schedule jobs
$1,0,2,5,4,6$), and $\BAD(G^1)=\BAD(G^2)=5$.

\end{document}